\documentclass[a4paper,USenglish,cleveref,autoref,thm-restate]{lipics-v2021}

\title{A Sound and Complete Characterization of\\ 
Fair Asynchronous Session Subtyping}
\author{Mario Bravetti}{University of Bologna}{mario.bravetti@unibo.it}{https://orcid.org/0000-0001-5193-2914}{}
\author{Luca Padovani}{University of Bologna}{luca.padovani2@unibo.it}
{https://orcid.org/0000-0001-9097-1297}{}
\author{Gianluigi Zavattaro}{University of Bologna}{gianluigi.zavattaro@unibo.it}{https://orcid.org/0000-0003-3313-6409}{}
\authorrunning{M. Bravetti, L. Padovani and G. Zavattaro}
\Copyright{Luca Padovani and Gianluigi Zavattaro}
\category{} 
\relatedversion{} 

\nolinenumbers

\EventEditors{John Q. Open and Joan R. Access}
\EventNoEds{2}
\EventLongTitle{42nd Conference on Very Important Topics (CVIT 2016)}
\EventShortTitle{CVIT 2016}
\EventAcronym{CVIT}
\EventYear{2016}
\EventDate{December 24--27, 2016}
\EventLocation{Little Whinging, United Kingdom}
\SeriesVolume{42}
\ArticleNo{23}

\usepackage{xspace}
\usepackage{xcolor}
\usepackage{amsthm}
\usepackage{stmaryrd}
\usepackage{mathtools}
\usepackage{mathpartir}
\usepackage{cleveref}
\usepackage{prooftree}
\usepackage{upgreek}
\usepackage{thm-restate}
\usepackage{tikz}
\usetikzlibrary{arrows,shapes,automata,spy,positioning,
  calc,shadows,fit,arrows.meta,patterns,backgrounds,decorations.pathreplacing}
\usepackage{tikz-qtree}

\newcommand{\F}{\mathcal{F}}

\newif\ifcomments
\commentstrue
\commentsfalse

\newcommand{\marginnote}[2]{%
  \ifcomments%
    $^{\color{magenta}\mathclap\star}$%
    \marginpar[
        \flushright\tiny\sf\textbf{#1}: #2
    ]{
        \flushleft\tiny\sf\textbf{#1}: #2
    }%
  \fi%
}
\newcommand{\LP}[1]{\marginnote{LP}{\color{blue}#1}}
\newcommand{\GZ}[1]{\marginnote{GZ}{\color{purple}#1}}
\newcommand{\MB}[1]{\marginnote{MB}{\color{red}#1}}

\definecolor{mc}{rgb}{0.5,0,0}

\newcommand{\set}[1]{\braces{#1}}
\newcommand{\Nat}{\mathbb{N}}

\newcommand{\rulename}[1]{\textnormal{\textup{\textsc{\small\bracks{#1}}}}}
\newcommand{\eoe}{\hfill$\lrcorner$}

\newcommand{\mkkeyword}[1]{\mathsf{\color{blue}#1}}

\newcommand{\proofcase}[1]{\textbf{Case #1.}\xspace}

\newcommand{\ie}{\emph{i.e.}\xspace}

\newcommand{\eg}{\emph{e.g.}\xspace}

\newcommand{\bracks}[1]{[#1]}
\newcommand{\braces}[1]{\{#1\}}

\newcommand{\pol}{p}
\newcommand{\inp}{{?}}
\newcommand{\out}{{!}}

\newcommand{\MessageType}{\MessageTypeS}
\newcommand{\MessageTypeS}{\TagA}
\newcommand{\MessageTypeT}{\TagB}

\newcommand{\SessionTypeT}

\newcommand{\End}{\mkkeyword{end}}

\newcommand{\Branch}{\sum}

\newcommand{\branch}{\mathbin{+}}

\newcommand{\Tag}[1][a]{\mathsf{\color{teal}#1}}
\newcommand{\TagA}{\Tag[a]}
\newcommand{\TagB}{\Tag[b]}
\newcommand{\TagC}{\Tag[c]}
\newcommand{\TagD}{\Tag[d]}

\newcommand{\TagTC}{\Tag[tc]}
\newcommand{\TagTM}{\Tag[tm]}
\newcommand{\TagDONE}{\Tag[done]}
\newcommand{\TagOVER}{\Tag[over]}

\newcommand{\Action}{\ActionA}
\newcommand{\ActionA}{\alpha}
\newcommand{\ActionB}{\beta}

\newcommand{\Actions}{\ActionsA}
\newcommand{\ActionsA}{\varphi}
\newcommand{\ActionsB}{\psi}

\newcommand{\session}[2]{#1\mathbin\|#2}

\newcommand{\tr}{\mathsf{tr}}

\newcommand{\actions}{\mathsf{act}}
\newcommand{\inputs}{\mathsf{inp}}
\newcommand{\outputs}{\mathsf{out}}
\newcommand{\co}[1]{\overline{#1}}

\newcommand{\srel}{\mathcal{S}}

\newcommand{\eqdef}{\stackrel{\smash{\textsf{\upshape\tiny def}}}=}

\newcommand{\conv}{\sqsubseteq}
\newcommand{\correct}[2]{#1\! \mathrel{\Join}\! #2}
\newcommand{\subt}{\preceq}

\newcommand{\red}{\rightarrow}

\newcommand{\wred}{\red^*}

\newcommand{\lred}[1]{\stackrel{#1}{\longrightarrow}}
\newcommand{\xlred}[1]{\xrightarrow{#1}}
\newcommand{\nlred}[1]{\longarrownot\lred{#1}}

\newcommand{\wlred}[1]{\stackrel{#1}{\Longrightarrow}}

\newcommand{\TagReq}{\Tag[req]}
\newcommand{\TagResp}{\Tag[resp]}
\newcommand{\TagStop}{\Tag[stop]}
\newcommand{\TagYes}{\Tag[yes]}
\newcommand{\TagNo}{\Tag[no]}

\newcommand{\msg}[1]{\Tag[#1]}
\newcommand{\snd}[1]{!{\msg{#1}}}

\newcommand{\exTM}{\msg{tm}}
\newcommand{\exTC}{\msg{tc}}
\newcommand{\exDONE}{\msg{done}}
\newcommand{\exOVER}{\msg{over}}

\newcommand{\exGR}{T'_G}
\newcommand{\exG}{T_G}
\newcommand{\exS}{T_S}

\newcommand{\exPSR}{T'_{PS}}
\newcommand{\exPS}{T_{PS}}
\newcommand{\exC}{T_{C}}

\newcommand{\rcv}[1]{?\msg{#1}}

\newcommand{\dual}[1]{\overline{#1}}

\tikzset{
  every state/.style={minimum size=1pt,inner sep=1.5pt, initial text={}},
  mycfsm/.style={
    font=\scriptsize,
    initial where=left,
    initial distance=0.25cm,
    ->,>=stealth,auto, node distance=0.8cm and 0.8cm,
    scale=1, every node/.style={transform shape},
    baseline=(current  bounding  box.center)
  },
  ogate/.style = {
    diamond, draw, fill=white,
    minimum size=4mm,
    inner sep=0pt,
    postaction={path picture={%
        \draw[black]
        ([yshift=\gatedistancein]path picture bounding box.south) -- ([yshift=-\gatedistancein]path picture bounding box.north)
        ([xshift=-\gatedistancein]path picture bounding box.east) -- ([xshift=\gatedistancein]path picture bounding box.west)
        ;}}, drop shadow},
  agate/.style={draw,rectangle,
    minimum size=3mm,
    inner sep=0pt,
    fill=white,
    postaction={path picture={%
        \draw[black]
        ([yshift=\gatedistanceinand]path picture bounding box.south) --
        ([yshift=-\gatedistanceinand]path picture bounding box.north) ;}}, drop shadow},
  source/.style={draw,circle,fill=white,
    minimum size=3mm,
    inner sep=0pt, drop shadow},
  sink/.style={draw,circle,double,fill=white,
    minimum size=3mm,
    inner sep=0pt, drop shadow},
  intera/.style = {rectangle, draw=black, align=center, fill=white, rounded corners=0.1cm,
    minimum height=12,
    inner sep=2pt, drop shadow},
  line/.style = {draw,->, rounded corners=0.07cm,>=latex},
  venn/.style={preaction={fill, #1},opacity=0.6},
  cnode/.style={rectangle,draw=black,inner sep=2pt},
  ancestor/.style={densely dashed,->},
  silentedge/.style={>=latex,->},
  nlabel/.style={fill=white,inner sep=0pt,font=\footnotesize},
  notexplo/.style={fill=gray!10},
  echnode/.style={rectangle,draw=black,inner sep=2pt},
  schnode/.style={diamond,draw=black,inner sep=0pt},
}

\begin{document}
\maketitle

\begin{abstract}
    Session types are abstractions of communication protocols enabling the
    static analysis of message-passing processes. Refinement notions for session
    types are key to support safe forms of process substitution while preserving
    their compatibility with the rest of the system.
    Recently, a fair refinement relation for asynchronous session types has been
    defined allowing the anticipation of message outputs with respect to an
    unbounded number of message inputs.
    This refinement is useful to capture common patterns in communication
    protocols that take advantage of asynchrony. However, while the semantic
    (\emph{\`a la testing}) definition of such refinement is straightforward,
    its characterization has proved to be quite challenging. In fact, only a
    sound but not complete characterization is known so far.
    In this paper we close this open problem by presenting a sound and complete
    characterization of asynchronous fair refinement for session types. We
    relate this characterization to those given in the literature for
    \emph{synchronous} session types by leveraging a novel labelled transition
    system of session types that embeds their asynchronous semantics.
\end{abstract}

\section{Introduction}
\label{sec:introduction}

\begin{figure}[t]
  \centering
  \begin{tabular}{c@{\qquad}c@{\qquad}c}
    \begin{tikzpicture}[mycfsm, node distance = 0.5cm and 0.9cm
      ,scale=1.2, every node/.style={transform shape}]
      \node[state, initial, initial where=left] (s0) {$0$};
      \node[state, right =of s0] (s1) {$1$};
      \node[state, right=of s1] (s2) {$2$};
      \path 
      (s0) edge [loop above] node {$\snd{\exTC}$} (s0)
      (s0) edge node {$\snd{\exDONE}$} (s1)
      (s1) edge [loop above] node [above] {$\rcv{\exTM}$} (s1)
      (s1) edge node [below] {$\rcv{\exOVER}$} (s2)
      ; 
    \end{tikzpicture}
    &
    \begin{tikzpicture}[mycfsm, node distance = 0.5cm and 0.9cm
      ,scale=1.2, every node/.style={transform shape}]
      \node[state, initial, initial where=left] (s0) {$0$};
      \node[state, right =of s0] (s1) {$1$};
      \node[state, right=of s1] (s2) {$2$};
      \path 
      (s0) edge [loop above] node [above] {$\rcv{\exTM}$} (s0)
      (s0) edge node [below] {$\rcv{\exOVER}$} (s1)
      (s1) edge [loop above] node {$\snd{\exTC}$} (s1)
      (s1) edge node {$\snd{\exDONE}$} (s2)
      ; 
    \end{tikzpicture}
    &
    \begin{tikzpicture}[mycfsm, node distance = 0.5cm and 0.9cm
      ,scale=1.2, every node/.style={transform shape}]
      \node[state, initial, initial where=left] (s0) {$0$};
      \node[state, right =of s0] (s1) {$1$};
      \node[state, right=of s1] (s2) {$2$};
      \path 
      (s0) edge [loop above] node [above] {$\snd{\exTM}$} (s0)
      (s0) edge node [below] {$\snd{\exOVER}$} (s1)
      (s1) edge [loop above] node {$\rcv{\exTC}$} (s1)
      (s1) edge node {$\rcv{\exDONE}$} (s2)
      ; 
    \end{tikzpicture}
    \\ 
    $\exGR$  & $\exG= \dual{\exS}$ & $\exS$  
  \end{tabular}
  \caption{Satellite protocols. $\exGR$ is the refined session type of the ground station, $\exG$ is the session type of ground station, and $\exS$ is the session type of the spacecraft.}
\label{fig:oldex-types}
\end{figure}
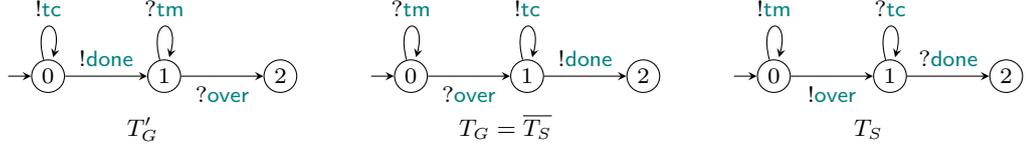


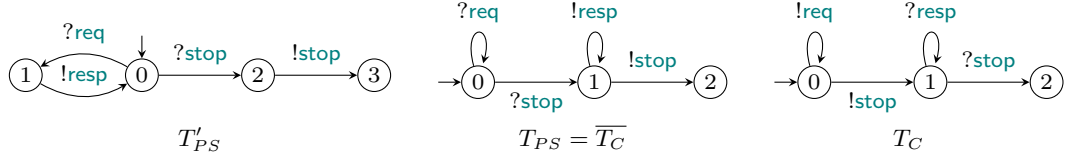
\begin{figure}[t]
  \centering
  \begin{tabular}{c@{\quad}c@{\quad}c} 
    \begin{tikzpicture}[mycfsm, node distance = 0.5cm and 0.9cm
      ,scale=1.2, every node/.style={transform shape}]
      \node[state, initial, initial where=above] (s0) {$0$};
      \node[state, left =of s0] (s1) {$1$};
      %
      \node[state, right =of s0] (s2) {$2$};
      \node[state, right =of s2] (s3) {$3$};
      \path
      (s0) edge [bend right] node [above] {$\rcv{\TagReq}$} (s1)
      (s0) edge node {$\rcv{\TagStop}$} (s2)
      (s1) edge [bend right]  node {$\snd{\TagResp}$} (s0)
      (s2) edge node {$\snd{\TagStop}$} (s3)   
      ; 
    \end{tikzpicture}
    &
    \begin{tikzpicture}[mycfsm, node distance = 0.5cm and 0.9cm
      ,scale=1.2, every node/.style={transform shape}]
      \node[state, initial, initial where=left] (s0) {$0$};
      \node[state, right =of s0] (s1) {$1$};
      \node[state, right=of s1] (s2) {$2$};
      \path 
      (s0) edge [loop above] node [above] {$\rcv{\TagReq}$} (s0)
      (s0) edge node [below] {$\rcv{\TagStop}$} (s1)
      (s1) edge [loop above] node {$\snd{\TagResp}$} (s1)
      (s1) edge node {$\snd{\TagStop}$} (s2)
      ; 
    \end{tikzpicture}
    &
    \begin{tikzpicture}[mycfsm, node distance = 0.5cm and 0.9cm
      ,scale=1.2, every node/.style={transform shape}]
      \node[state, initial, initial where=left] (s0) {$0$};
      \node[state, right =of s0] (s1) {$1$};
      \node[state, right=of s1] (s2) {$2$};
      \path 
      (s0) edge [loop above] node [above] {$\snd{\TagReq}$} (s0)
      (s0) edge node [below] {$\snd{\TagStop}$} (s1)
      (s1) edge [loop above] node {$\rcv{\TagResp}$} (s1)
      (s1) edge node {$\rcv{\TagStop}$} (s2)
      ; 
    \end{tikzpicture}
    \\ 
    $\exPSR$  & $\exPS = \dual{\exC}$ & $\exC$ 
\end{tabular}
\caption{Stream processing server $\exPSR$ is the refined session type of the batch processing server $\exPS$, and $\exC$ is the session type of the client.}
\label{fig:runex-types}
\end{figure}

Abstract models such as communicating finite-state machines~\cite{BZ83} and
asynchronous session types~\cite{HYC16} are essential to reason about
correctness of distributed systems whose components communicate with
point-to-point \textsc{fifo} channels.
%
A fundamental issue, which makes it possible to manage system correctness in a
compositional way, is the development of techniques allowing a component to be
\emph{refined} independently of the others, without compromising the correctness
of the whole system.
In this respect, the notion of fair asynchronous refinement for session types
introduced by Bravetti, Lange and Zavattaro~\cite{BravettiLangeZavattaro24}
guarantees fair termination and implies all the desirable safety and liveness
properties of communicating systems, including communication safety, deadlock
freedom, absence of orphan messages, and lock freedom. However, while the
semantic (\emph{\`a la testing}) definition of such refinement is
straightforward, its (coinductive) characterization (\emph{\`a la session
subtyping}) has proved to be quite challenging. In fact, Bravetti, Lange and
Zavattaro~\cite{BravettiLangeZavattaro24} only provide a sound but not complete
characterization.

With respect to previous notions of asynchronous session subtyping~\cite{ESOP09,
CDY2014, MariangiolaPreciness} (which leverage asynchrony by allowing the
refined component to anticipate message emissions, but only under certain
conditions) fair asynchronous subtyping~\cite{BravettiLangeZavattaro24} makes it
possible to encompass subtypes that occur naturally in communication protocols,
\eg where two parties simultaneously send each other a finite but unspecified
amount of messages before consuming them from their buffers. We illustrate this
scenario in Figure~\ref{fig:oldex-types}, which depicts the interaction between
a spacecraft $S$ and a ground station $G$ that communicate via two unbounded
asynchronous channels (one in each direction).
For convenience, the protocols are represented as communicating finite-state
machines~\cite{BZ83}, \ie with a labeled transition system-like notation where
``?'' represents inputs, ``!'' represents outputs, and the initial state is
indicated by an incoming arrow.
Consider $\exS$ and $\exG$ first.
Session type $\exS$ is the abstraction of the spacecraft, which may send a
finite but unspecified number of telemetries $\exTM$ (abstractly modeling a
\texttt{for} loop), followed by a message $\exOVER$.
In the second phase, the spacecraft receives a number of telecommands ($\exTC$),
followed by a message $\exDONE$.
%
%
In principle, the ground station should behave as the type $\exG$ which is the
dual of $\exS$ (denoted by $\dual{\exS}$) where outputs become inputs and
vice-versa.
  However, since the flyby window may be short, it makes sense to implement the
  ground station so that it communicates with the satellite 
  by anticipating the output of the commands.\GZ{Ho tolto "full duplex" perché
  sopra diciamo che ci sono due canali nelle due direzioni} That is, the ground
  station follows the type
$\exGR$, which sends telecommands before receiving telemetries.
In this way $\exGR$ and $\exS$ interact in a symmetric manner:
they first send all of their messages and then consume the messages 
sent from the other partner. No communication error can occur, and the 
communication protocol can always terminate successfully, with empty queues.
In fact $\exGR$ and $\exS$ also form a correct composition in the sense
that every message that is sent by one participant is eventually
received by the other.
The session subtyping presented by Bravetti, Lange and
Zavattaro~\cite{BravettiLangeZavattaro24}, which is proved to be a sound
characterization of fair asynchronous session refiment, makes it possible to
actually \emph{prove} that $\exGR$ refines $\exG$.

Let us now consider, in Figure~\ref{fig:runex-types}, a more complex scenario
concerning a processing server $PS$ and its client $C$. Consider $\exPS$ and
$\exC$ first.
Session type $\exPS$ is the abstraction of a batch processing server and $\exC$
that of a client. Note that these types are respectively isomorphic to $\exG$
and $\exS$ of Figure~\ref{fig:oldex-types}: now the client $\exC$ sends generic
requests $\TagReq$ (instead of specific telemetry data $\exTM$) and, when it
decides to stop sending via $\TagStop$, it keeps waiting for responses
$\TagResp$ (instead of telecommand data $\exTC$) until it receives $\TagStop$.
As in the previous scenario, the client $\exC$ and the batch processing server
$\exPS$ (which is its dual) obviously form a correct composition. Also here it
is possible to consider a more efficient version of the processing server, \ie a
stream processing server $\exPSR$ which immediately (and asynchronously) sends
the response to each request it receives. In this scenario, it may be natural to
send responses one at a time after each request.
In fact also $\exPSR$ and $\exC$ form a correct composition. Actually, we have
that $\exPSR$ is a fair asynchronous session refinement of $\exPS$, meaning that
$\exPSR$ is a good replacement of $\exPS$ in any possible context. However, the
characterization provided by Bravetti, Lange and
Zavattaro~\cite{BravettiLangeZavattaro24} is unable to capture this specific
refinement. 


Technically speaking, the coinductive characterization of subtyping given by
Bravetti, Lange and Zavattaro \cite{BravettiLangeZavattaro24} is not complete
because it does not support \emph{output covariance}, that is the possibility
for a subtype to remove branches in output choices, corresponding to a reduced
internal nondeterminism.
This feature is needed to relate $\exPSR$ and $\exPS$ of
Figure~\ref{fig:runex-types} because the state $1$ (resp. $2$) of $\exPSR$, in
which the forced internal choice of emitting $\snd{\TagResp}$ (resp.
$\snd{\TagStop}$) is taken, has a unique corresponding state in $\exPS$, namely
state $1$, where two possible choices are present (either $\snd{\TagResp}$ or
$\snd{\TagStop}$ can be sent).

In this paper we close the open problem of finding a sound and complete
characterization of asynchronous fair refinement for session types. The
characterization we present is crucially based on a novel asynchronous semantics
for session types that: (i) expresses their asynchronous behavior without
explicit modeling of FIFO buffers (as in the original paper of Bravetti, Lange
and Zavattaro~\cite{BravettiLangeZavattaro24}), and (ii) allows us to adapt, to
the asynchronous case, the approach introduced by Ciccone and
Padovani~\cite{Padovani16,CicconePadovani22B} for providing complete
characterizations of \emph{fair session subtyping} in the \emph{synchronous}
case. Being complete, our new subtyping allows us to prove that the stream
processing server $\exPSR$ is a refinement of the batch processing server
$\exPS$.

\subparagraph*{Structure of the paper.}
In \cref{sec:preliminaries} we recall the notion of fair asynchronous refinement
and the buffer-based asynchronous semantics of session types given by Bravetti,
Lange and Zavattaro~\cite{BravettiLangeZavattaro24}.
In \cref{sec:session-types} we present our novel alternative asynchronous
semantics and prove that it is equivalent to the previous one, in the sense that
it characterizes the same notion of session composition correctness.
In \cref{sec:subtyping} we present a sound and complete characterization of
asynchronous fair refinement based on the semantics of \Cref{sec:session-types}.
We conclude in \cref{sec:conclusion} with a more detailed description of related
work and hints at further developments.
Because of space constraints, proofs and additional technical material have been
postponed in the appendices.

\newcommand{\conf}{s}
\newcommand{\parconf}{|}
\newcommand{\cnfg}[4]{[#1,#2]\parconf[#3,#4]}
\newcommand{\word}{\omega}
\newcommand{\Tbranchindex}[4]{\&\{{#1}_{#3}:{#2}_{#3}\}_{#3\in #4}}
\newcommand{\TbranchindexNoidx}[5]{\&\{{#1}_{#3}:{#2}, \ #5 \}_{ #4}}
\newcommand{\Tbranch}[2]{\Tbranchindex{#1}{#2}{i}{I}}
\newcommand{\Tbranchsingledec}[3]{\&^{#3}\{{#1}:{#2}\}}
\newcommand{\Tbranchsingle}[2]{\&\{{#1}:{#2}\}}
\newcommand{\Tbranchset}[3]{\&\{{#1}\!:\!{#2}\}_{#1\in #3}}
\newcommand{\TbranchK}[2]{k\Tbranchindex{#1}{#2}{i}{I}}
\newcommand{\TbranchindexM}[5]{\&\{{#1}_{#3}(#5_{#3}):{#2}_{#3}\}_{#3\in #4}}
\newcommand{\TbranchM}[3]{\TbranchindexM{#1}{#2}{i}{I}{#3}}
\newcommand{\Tbranchsimple}[2]{\&\{{#1}:{#2}\}}
\newcommand{\Tbra}[1]{\&\{#1\}}
\newcommand{\Tselectindex}[4]{\oplus\{{#1}_{#3}:{#2}_{#3}\}_{#3\in #4}}
\newcommand{\Tselect}[2]{\Tselectindex{#1}{#2}{i}{I}}
\newcommand{\Tselectsingle}[2]{\oplus\{{#1}:{#2}\}}
\newcommand{\Tselectset}[3]{\oplus\{{#1}:{#2}\}_{#1\in #3}}
\newcommand{\Tsel}[1]{\oplus\{ #1 \}}
\newcommand{\TselectindexNoidx}[5]{\oplus\{{#1}_{#3}:{#2}, \ #5 \}_{ #4}}
\newcommand{\TselectK}[2]{k\Tselectindex{#1}{#2}{i}{I}}
\newcommand{\TselectindexM}[5]{\oplus\{{#1}_{#3}\langle #5_{#3}\rangle:{#2}_{#3}\}_{#3\in #4}}
\newcommand{\TselectM}[3]{\TselectindexM{#1}{#2}{i}{I}{#3}}
\newcommand{\Tselectsimple}[2]{\oplus\{{#1}:{#2}\}}
\newcommand{\append}{\!\cdot\!}
\newcommand{\unf}[1]{\mathsf{unfold}(#1)}
\newcommand{\Tend}{\mathbf{end}}

\newcommand{\refine}{\subt}

\newcommand{\TagTask}{\Tag[task]}
\newcommand{\TagRes}{\Tag[res]}
\newcommand{\TagCommand}{\Tag[cmd]}
\newcommand{\TagData}{\Tag[data]}

\section{Preliminaries}
\label{sec:preliminaries}

We start by recalling the syntax and asynchronous semantics of session types as
well as the notion of fair refinement given by Bravetti, Lange and
Zavattaro~\cite{BravettiLangeZavattaro24}. In the definition of the syntax,
instead of using the $\oplus$, $\&$, and the $\mathit{rec}$ operators, we
consider an equivalent process algebraic notation with choice and prefix
\cite{Mil89} and a coinductive syntax to deal with possibly infinite session
types. Following Bravetti, Lange and
Zavattaro~\cite{BravettiLangeZavattaro24} we restrict to first-order session 
types, \ie~session types in which send/receive operations are used to 
exchange only messages of elementary \emph{singleton type}. Considering 
higher-order types, where send/receive operations can be used to
exchange also sessions,
would not modify significantly the proofs of our results. 

We assume the existence of a set of \emph{singleton types} ranged over by
$\TagA$, $\TagB$, $\dots$ that we use to label branches in session types. The
only value of type $\Tag$ is called \emph{tag} and is denoted by the same symbol
$\Tag$.
\emph{Pre-session types} are the possibly infinite trees coinductively generated
by the productions below:
\medskip
\[
    \textstyle
    \begin{array}{rr@{~}c@{~}l}
        \textbf{Polarity} & \pol & ::= & \inp \mid \out \\
        \textbf{Pre-session type} & S, T & ::= & \End \mid \sum_{i\in I} \pol\Tag_i.S_i
    \end{array}
\]
\medskip

The \emph{polarities} $\inp$ and $\out$ are used to distinguish \emph{input
actions} from \emph{output actions}.
A \emph{session type} is a pre-session type that satisfies the following
well-formedness conditions:
\begin{enumerate}
    \item\label{wf:branch} in every subtree of the form $\sum_{i\in I}
    \pol\Tag_i.S_i$, $I$ is a non-empty finite set and if $I$ contains more than
    one element, then the message types $\Tag_i$ are pairwise distinct tags;

    \item\label{wf:regular} the tree is \emph{regular}, namely it is made of
    finitely many distinct subtrees $\sum_{i\in I}
    \pol\Tag_i.S_i$. 
Courcelle~\cite{Courcelle83} established
    that every such tree can be expressed as the unique solution of a finite
    system of equations $\set{S_i = T_i}_{i\in I}$ where each metavariable $S_i$
    may occur (guarded) in the $T_j$'s;
\item\label{wf:end} every subtree $\sum_{i\in I}
    \pol\Tag_i.S_i$ contains a leaf of the form $\End$.
\end{enumerate}

A session type $\End$ describes a terminated session (\ie the corresponding
channel is no longer usable). 
%
A branching session type $\sum_{i\in I} \pol\Tag_i.S_i$ describes a channel used
for sending or receiving a message of type $\Tag_i$ and then according to $S_i$.
The well-formedness condition \eqref{wf:branch} above requires that message
types must be pairwise distinct tags if there is more than one branch.
Sometimes we write $\pol\Tag_1.S_1 + \cdots + \pol\Tag_n.S_n$ for $\sum_{i=1}^n
\pol\Tag_i.S_i$. Note that the polarity $\pol$ is the same in every branch, \ie
mixed choices are disallowed as usual.

The well-formedness condition \eqref{wf:regular} guarantees that we deal
only with session types representing trees consisting of finitely many
distinct subtrees as those that can be represented with a finite syntax 
and the $\mathit{rec}$ operator~\cite{Courcelle83}.

The well-formedness condition \eqref{wf:end} ensures that session types describe
protocols that can always terminate. This assumption is not usually present in
other session type syntax definitions where it is possible to also express
non-terminating protocols. As we will see in the following
(\Cref{def:compliance}), non-terminating protocols are not inhabited in our
theory. In particular, we consider a notion of correct type composition such
that, whatever sequence of interactions is performed, such sequence can always
be extended to reach successful session termination where both types
become $\End$. In conclusion, ruling out non-terminating protocols in the syntax
of session types does not affect the family of protocols we are interested in
modeling and at the same time simplifies the technical development.

It is worth to mention that Bravetti, Lange and
Zavattaro~\cite{BravettiLangeZavattaro24} do not consider the well formedness
condition \eqref{wf:end}, hence allowing for the specification
of non-terminating types. In the comparison with the related literature in \Cref{sec:conclusion}
we discuss the impact of the absence of this condition on the definition
of the fair asynchronous subtyping reported in that paper.

We write $\co{S}$ for the \emph{dual} of $S$, which is the session type
corecursively defined by the equations
\[
    \textstyle
    \co{\End} = \End
    \qquad
    \co{\sum_{i\in I} \pol\Tag_i.S_i} = \sum_{i\in I} \co\pol\Tag_i.\co{S_i}
\]
where $\co\pol$ is the dual of the polarity $\pol$ so that $\co\inp = \out$ and
$\co\out = \inp$. As usual, duality affects the direction of messages but not
message types.

We now define the transition system that we use to formalize the notion of
correct session type composition and the induced notion of
refinement~\cite{BravettiLangeZavattaro24}.

The transition system makes use of explicit FIFO queues of messages that have
been sent asynchronously. A \emph{configuration} is a term
$\cnfg{S}{\word_S}{T}{\word_T}$ where $S$ and $T$ are session types equipped
with two queues $\word_S$ and $\word_T$ of incoming messages. We write
$\epsilon$ for the empty queue and we use $\conf$, $\conf'$, etc. to range over
configurations.

\begin{definition}[Transition Relation \cite{BravettiLangeZavattaro24}]
\label{def:transrel}
The transition relation $\rightarrow$ over configurations 
is the minimal relation satisfying the rules below (plus symmetric ones,
omitted):
\begin{enumerate}
\item \label{it:trans-send}
if $j \in I$ then
$\cnfg{\sum_{i\in I} !\Tag_i.S_i}{\word_S}{T}{\word_T} \rightarrow
\cnfg{S_j}{\word_S}{T}{\word_T \Tag_j}$;
\item \label{it:trans-rcv}
if $j \in I$ then
$\cnfg{\sum_{i\in I} ?\Tag_i.S_i}{\Tag_j \word_S}{T}{\word_T} \rightarrow
\cnfg{S_j}{\word_S}{T}{\word_T}$. 
\end{enumerate}
We write $\rightarrow^*$ for the reflexive and transitive closure of
the $\rightarrow$ relation.
\end{definition}
%

\begin{example}\label{ex:diff_semantics}
Consider types 
$S=\out\TagA.\inp\TagA.\End + \out\TagB.\inp\TagB.\End$
and
$T=\out\TagA.\inp\TagA.\End$.
The configuration $\cnfg{S}{\epsilon}{T}{\epsilon}$
has the sequence of transitions
$\cnfg{S}{\epsilon}{T}{\epsilon} \! \rightarrow 
 \cnfg{\inp\TagA.\End}{\epsilon}{T}{\TagA} \! \rightarrow
 \cnfg{\inp\TagA.\End}{\TagA}{\inp\TagA.\End}{\TagA}$ $\rightarrow
 \cnfg{\inp\TagA.\End}{\TagA}{\End}{\epsilon} \rightarrow
 \cnfg{\End}{\epsilon}{\End}{\epsilon}
$,
which terminates in a configuration with both types equal to $\End$ 
and with both message queues empty.
It also has, e.g., another computation 
$\cnfg{S}{\epsilon}{T}{\epsilon} \rightarrow
 \cnfg{S}{\TagA}{\inp\TagA.\End}{\epsilon} \rightarrow
 \cnfg{\inp\TagB.\End}{\TagA}{\inp\TagA.\End}{\TagB}$
which terminates in a configuration with both types different from $\End$ 
and with both message queues not empty.
\eoe
\end{example}

Intuitively, the composition of the two types $S$ and $T$ in
\cref{ex:diff_semantics} is \emph{incorrect} because it can lead to a deadlocked
configuration $\cnfg{\inp\TagB.\End}{\TagA}{\inp\TagA.\End}{\TagB}$ discussed
therein. Following Bravetti, Lange and Zavattaro~\cite{BravettiLangeZavattaro24}
we formalize the correct composition of two types as a \emph{compliance}
relation.


\begin{definition}[Compliance \cite{BravettiLangeZavattaro24}]
\label{def:compliance} \label{def:compatibility} Given a configuration $\conf$
we say that it is a correct composition if, whenever $\conf\rightarrow^*\conf'$,
then also $ \conf' \rightarrow^\ast \cnfg{\End}{\epsilon}{\End}{\epsilon}$.
\noindent
Two session types $S$ and $T$ are \emph{compliant} if
$\cnfg{S}{\epsilon}{T}{\epsilon}$ is a correct composition.
\end{definition}

We can now formally state that the types $S$ and $T$ of \cref{ex:diff_semantics}
are not \emph{compliant} because of the sequence of transitions discussed at the
end of \cref{ex:diff_semantics}: $\cnfg{S}{\epsilon}{T}{\epsilon} \rightarrow
\cnfg{S}{\TagA}{\inp\TagA.\End}{\epsilon} \rightarrow
\cnfg{\inp\TagB.\End}{\TagA}{\inp\TagA.\End}{\TagB}$ leading to a configuration
different from the successfully terminated computation
$\cnfg{\End}{\epsilon}{\End}{\epsilon}$, and without outgoing transitions.

Compliance induces a semantic notion of type refinement, as follows.

\begin{definition}[Refinement \cite{BravettiLangeZavattaro24}]\label{def:refine}
  A session type $S$ refines $T$, written $S \refine T$, if, for every
  $R$ s.t. $T$ and $R$ are compliant, then $S$ and $R$ are also
  compliant.
\end{definition}

In words, this definition says that a process behaving as $T$ can be ``safely''
replaced by a process behaving as $S$ when $S$ is a refinement of $T$. Indeed,
the peer process, which is assumed to behave according to some session type $R$
such that $R$ and $T$ are compliant, will still interact correctly after the
substitution has taken place.
We quote ``safely'' above since we want to stress that the notion of session
correctness being preserved here is not merely a \emph{safety property}, but
rather a combination of a safety property (if the interaction gets stuck, then
it has successfully terminated) and a liveness property (the interaction can
always successfully terminate).

Because of the universal quantification on the type $R$, \Cref{def:refine}
provides few insights about what it means for $S$ to be a refinement of $T$. For
this reason, it is important to provide alternative (but equivalent)
characterizations of type refinement. Bravetti, Lange and
Zavattaro~\cite{BravettiLangeZavattaro24} present a characterization
of refinement defined coinductively \emph{à la session subtyping} \cite{GayHole05},
hence called fair asynchronous subtyping,
which can be used to prove some interesting and non trivial
relations as those discussed in \Cref{sec:introduction}.

\begin{example}[Ground station and satellite communication \cite{BravettiLangeZavattaro24}]
\label{ex:satellite}
We now present the session types describing the possible
behaviours of the ground station in the satelllite communication 
example presented in \Cref{sec:introduction}.
We use type $T$ to represent the communication behaviour
of the automaton $\exG$ in \Cref{fig:oldex-types}:
$T = \inp\TagTM.T +
\inp\TagOVER.T'$ where $T' = \out\TagTC.T' + \out\TagDONE.\End$.
We use $S$ for the automaton $\exGR$:
$S =
\out\TagTC.S + \out\TagDONE.S'$ where $S' = \inp\TagTM.S' + \inp\TagOVER.\End$.
Type $T$   
specifies the initially expected behaviour of the ground station 
as a complementary protocol w.r.t. to the one,  $\exS$, followed by 
the spacecraft. Type $S$ specifies an implementation where the 
emission of the telecommands is anticipated w.r.t. the consumption 
of the received telemetries. The subtyping of \cite{BravettiLangeZavattaro24}
allows to prove that $S \subt T$, hence we can conclude that the 
protocol with the anticipation of the telemetries  
correctly implements (refines) the expected communication behaviour of the
ground station. 
\eoe
\end {example}

The coinductive subtyping in~\cite{BravettiLangeZavattaro24} 
is a sound but not complete characterization of fair asynchronous refinement. 
In particular, it is uncapable of establishing relations $S \subt T$ where $S$
describes a more deterministic protocol than $T$ (output covariance). We have seen an instance of
this case in \cref{sec:introduction} when discussing the stream and batch
processing servers. Liveness-preserving refinements are difficult to
characterize because they are intrinsically non-local: the removal of an output
action in a region of a session type may compromise the successful termination
of the protocol along a branch that is ``far away'' from the location where the
output was removed. In our case this phenomenon is further aggravated by the
presence of asynchrony, which allows output actions to be anticipated and
therefore to be moved around.

In order to provide a complete characterization of refinement, we use an
alternative, queue-less semantics of asynchronous session types that we
introduce in the next section.

\section{An Alternative Asynchronous Semantics for Session Types}
\label{sec:session-types}

The first step towards our characterization of fair refinement is the definition
of a novel asynchronous semantics for session types that does not make use of
explicit queues. This alternative semantics will allow us to adapt the complete
characterization of fair synchronous subtyping for session
types~\cite{Padovani16,CicconePadovani22B} to the asynchronous case. 


%

We now introduce a \emph{labelled transition system} to describe the sequences
of input/output actions that are allowed by a session type. A \emph{label} is a
pair made of a \emph{polarity} and a \emph{tag} $\Tag$:
\[
    \textbf{Label}\quad
    \ActionA, \ActionB ::= \pol\Tag
    \qquad\qquad
\]

We will refer to the dual of a label $\Action$ with the notation $\co\Action$,
where $\co{\pol\Tag} \eqdef \co\pol\Tag$.

\newcommand{\ChannelTrans}{\rulename{l-chn}\xspace}
\newcommand{\TagTrans}{\rulename{l-tag}\xspace}
\newcommand{\ChannelTransAsync}{\rulename{l-chn-async}\xspace}
\newcommand{\TagTransAsync}{\rulename{l-tag-async}\xspace}
\newcommand{\SyncTrans}{\rulename{l-sync}\xspace}
\newcommand{\AsyncTrans}{\rulename{l-async}\xspace}

The labelled transition system (LTS) of asychronous session types is
\emph{coinductively} defined thus:
\begin{mathpar}
    \inferrule[\SyncTrans]{~}{
        \textstyle
        \sum_{i\in I} \pol\Tag_i.S_i \lred{\pol\Tag_k} S_k
    }
    ~k\in I
    \and
    \inferrule[\AsyncTrans]{
        \forall i\in I: S_i \lred{\inp\Tag} T_i
    }{
        \textstyle
        \sum_{i\in I} \out\Tag_i.S_i \lred{\inp\Tag} \sum_{i\in I} \out\Tag_i.T_i
    }
\end{mathpar}

The rule \SyncTrans simply expresses the fact that a session type allows the
input/output action described by its topmost prefix, as expected.
The rule \AsyncTrans is concerned with asynchrony and lifts input actions that
are enabled deep in a session type, provided that such actions are only prefixed
by output actions.
The rationale for such ``asynchronous transitions''
is that the topmost outputs allowed by a session type may have
been performed asynchronously, so that the underlying input is enabled even if
the output messages have not been consumed yet.
For example, the session type $\out\TagA.\inp\TagB.S$ may evolve either
as
\begin{equation}
    \label{eq:diamond}
    \out\TagA.\inp\TagB.S \lred{\out\TagA} \inp\TagB.S \lred{\inp\TagB} S
    \text{\qquad or as\qquad}
    \out\TagA.\inp\TagB.S \lred{\inp\TagB} \out\TagA.S \lred{\out\TagA} S
\end{equation}

The first transition sequence is the expected one and describes an evolution in
which the order of actions is consistent with the syntactic structure of the
session type. The second transition sequence is peculiar to the asynchronous
setting and describes an evolution in which the input $\inp\TagB$ \emph{is
performed before the output $\out\TagA$}.
Asynchronous transitions seem to clash with the very nature of session types,
whose purpose is to impose an order on the actions performed by a process on a
channel. To better understand them, we find it useful to appeal to the usual
interpretation of labelled transition systems as descriptions of those actions
that trigger interactions with the environment: a process that behaves according
to the session type $\out\TagA.\inp\TagB.S$ is still required to send $\TagA$
before it starts waiting for $\TagB$. However, it will be capable of receiving
$\TagB$ before $\TagA$ has been consumed by the party with which it is
interacting, because $\TagA$ may have been queued. In this sense, we see the
queue as an entity associated with the process and not really part of the
environment in which the process executes.\LP{Questo testo è molto simile a
quello dell'articolo ECOOP.}\GZ{Secondo me possiamo lasciare così}

It is useful to introduce some additional notation related to actions and
transitions of session types.
We let $\ActionsA$ and $\ActionsB$ range over strings of labels and use
$\varepsilon$ to denote the empty string.
We write $\xlred{\Action_1\cdots\Action_n}$ for the composition
$\lred{\Action_1}\cdots\lred{\Action_n}$, we write $S \lred\Actions$ if $S
\lred\Actions T$ for some $T$ and we write $S \nlred\Actions$ if not $S
\lred\Actions$.

Two subtle aspects of the LTS deserve further comments.
First of all, asynchronous transitions are possible only provided that they are
enabled \emph{along every branch} of a session type. For example, if $S =
\out\TagA.\inp\TagC.S_1 + \out\TagB.(\inp\TagC.S_2 + \inp\TagD.S_3)$, then we
have
\[
    S \lred{\inp\TagC} \out\TagA.S_1 + \out\TagB.S_2
    \text{\qquad and\qquad}
    S \nlred{\inp\TagD}
    \text{\qquad and\qquad}
    S \xlred{\out\TagB\inp\TagD} S_3
\]

The $\inp\TagC$-labelled transition is enabled because that input action is
available regardless of the message (either $\TagA$ or $\TagB$) that has been
sent. On the contrary, the $\inp\TagD$-labelled transition is initially
disabled because the input of $\TagD$ is possible only if $\TagB$ has been
sent. So, the peer interacting with a process that follows the session type $S$
\emph{must} necessarily consume either $\TagA$ or $\TagB$ in order to understand
whether or not it can send $\TagD$.

Finally, we have to motivate the \emph{coinductive} definition of the LTS. This
aspect is related to both asynchrony and fairness. Consider a session type $S$
such that $S = \out\TagA.S + \out\TagB.\inp\TagC.T$. The question is whether or
not the $\inp\TagC$-labelled transition should be enabled from $S$. In
principle, $S$ allows the output of an infinite sequence of $\TagA$ messages and
so the $\inp\TagC$-labelled transition should be disabled. In this work,
however, we make the assumption that this behavior is unrealistic or
\emph{unfair}, therefore we want the $\inp\TagC$-labelled transition to be
enabled.
From a technical standpoint, the only derivation that allows us to express an
$\inp\TagC$-labelled transition for $S$ is the following one
\[
    \begin{prooftree}
        \[
            \mathstrut\smash\vdots
            \justifies
            S \lred{\inp\TagC} S'
            \using\AsyncTrans
        \]
        \[
            \justifies
            \inp\TagC.T \lred{\inp\TagC} T
            \using\SyncTrans
        \]
        \justifies
        S \lred{\inp\TagC} S'
        \using\AsyncTrans
    \end{prooftree}
\]
where $S' = \out\TagA.S' + \out\TagB.T$. This infinite derivation is legal only
if the LTS is coinductively defined.

The coinductive definition of the LTS may cast doubts on the significance of the
labels of transitions, in the sense that some (input) transitions could be
coinductively derivable for a given session type $S$ even if they do not
correspond to actions that are actually described within $S$.
For example, for the pre-session type $S = \out\Tag.S$ it would be possible to
derive \emph{every} transition of the form $S \lred{\inp\TagB} T$ by using
infinitely many applications of \AsyncTrans.
The following result rules out this possibility. The key element of the proof is
the fact that, because of the well-formedness condition \eqref{wf:end}, every
subtree of a session type contains a leaf of the form $\End$.

\begin{restatable}{proposition}{propcoindlts}
    \label{prop:coind-lts}
    If\/ $S \lred\Action$, then there exists $T$ and $\Actions$ made of output
    actions only such that $S \lred\Actions T \lred\Action$ and the last
    transition is derived by \SyncTrans.
\end{restatable}

After looking at the transition sequences in \eqref{eq:diamond}, the reader may
also wonder whether the LTS always satisfies the diamond property when a session
type simultaneously enables both input and output actions. Crucially, this is
the case.

\begin{proposition}
    \label{prop:diamond}
    If\/ $S \lred{\inp\TagA} S'$ and $S \lred{\out\TagB} S''$,
    then there exists $T$ such that $S' \lred{\out\TagB} T$ and $S''
    \lred{\inp\TagA} T$.
\end{proposition}

Having established these basic facts about the LTS of session types, we introduce
some more notation.
If $\Actions = \Action_1\cdots\Action_n$, we write $\Actions.S$ for
$\Action_1\dots\Action_n.S$. Note that $\varepsilon.S = S$.
Then, we define a notion of \emph{partial derivative} for session types that
resembles Brzozowski's derivative for regular expressions~\cite{Brzozowski64}.
In particular, if $S \lred\Actions T$ we say that $T$ is the derivative of $S$
after $\Actions$ and we write $S(\Actions)$ for $T$. Note that $S(\Actions)$ is
uniquely defined because of the well-formedness condition~\eqref{wf:branch}.
We define the \emph{inputs} and \emph{outputs} of a session type as $\inputs(S)
\eqdef \set{ \Tag \mid S \lred{\inp\Tag} }$ and $\outputs(S)
\eqdef \set{ \Tag \mid S \lred{\out\Tag} }$.
Note that $\inputs(\End) = \outputs(\End) = \emptyset$.
Occasionally we use $\outputs(\cdot)$ also as a predicate so that $\outputs(S)$
holds if and only if $\outputs(S) \ne \emptyset$.
Finally, the set of \emph{traces} of a session type is defined as $\tr(S) \eqdef
\set{ \Actions \mid S \lred\Actions \End }$.

\begin{example}
    \label{ex:stream-batch-types}
    Let us formalize the protocols of the stream and batch processing servers we
    have sketched in \Cref{sec:introduction}. Consider the session types
    \[
        S = \inp\TagReq.\out\TagResp.S + \inp\TagStop.\out\TagStop.\End
        \text{\qquad and\qquad}
        \begin{array}{r@{~}l}
            T = & \inp\TagReq.T + \inp\TagStop.T' \\
            T' = & \out\TagResp.T' + \out\TagStop.\End \\
        \end{array}
    \]
    and observe that $\tr(T) = \set{
    (\inp\TagReq)^m\inp\TagStop(\out\TagResp)^n\out\TagStop \mid m,n\in\Nat }$.
    Note also that $S \xlred{(\inp\TagReq)^n\inp\TagStop}
    \out\TagResp^n.\out\TagStop.\End$ for every $n\in\Nat$.
    Therefore, $S$ allows all the sequences of input actions allowed by $T$,
    after which it performs a subset of the sequences of output actions allowed
    by $T$. At the same time, $S$ allows the anticipation of output actions
    allowed by $T$.
    %
    \eoe
\end{example}

We now define a notion of correctness that adapts, to the new semantics, the
notion of correct composition used in \cref{def:compliance}.
A \emph{session composition} is a pair $\session{S}{T}$ of session types.
Session compositions reduce according to the following rule:
\begin{equation}
  \inferrule{
    S \lred{\co\Action} S'
    \\
    T \lred\Action T'
  }{
    \session{S}{T} \red \session{S'}{T'}
  }
  ~
  \begin{array}{l}
    \outputs(S) \subseteq \inputs(T)
    \\
    \outputs(T) \subseteq \inputs(S)
  \end{array}
\label{eq:red}
\end{equation}

Notice that we overload $\rightarrow$ and $\rightarrow^*$, which were previously
used to denote the transitions over session configurations in
\cref{def:transrel}. Their actual meaning is made clear by the context.
Moreover, note that a session composition is \emph{stuck} (\ie it does not
reduce) if one of the two session types in the composition enables an output for
which the corresponding input is disabled in the other session type. 

\begin{example}\label{ex:diff_semantics_2} Consider the types
$S=\out\TagA.\inp\TagA.\End + \out\TagB.\inp\TagB.\End$ and
$T=\out\TagA.\inp\TagA.\End$ already discussed in \cref{ex:diff_semantics}. We
have observed that their configuration $\cnfg{S}{\epsilon}{T}{\epsilon}$ can
perform transitions according to the relation in \cref{def:transrel}. On the
contrary, the session composition $\session{S}{T}$ is stuck because $\TagB \in
\outputs(S)$ while $\TagB \not\in \inputs(T)$, hence $\outputs(S) \not\subseteq
\inputs(T)$.
\eoe
\end{example}

The definition of correct session composition is the same as
\cref{def:compliance} adapted to the new asynchronous semantics 
of types. Intuitively, in a correct session composition $\session{S}{T}$, 
the only state that is allowed to be stuck is the final one, where both 
session types have reduced to
$\End$ meaning that the session has successfully terminated. Moreover, every
intermediate state $\session{S'}{T'}$ that is reachable from $\session{S}{T}$
must itself be on a path that leads to successful termination of the session. 

\begin{definition}
  \label{def:correctness}
  We say that $\session{S}{T}$ is \emph{correct}, notation $\correct{S}{T}$, if
  $\session{S}{T} \wred \session{S'}{T'}$ implies that
  $\session{S'}{T'} \wred \session{\End}{\End}$.
\end{definition}

A simple example of correct session composition is
$\session{\out\TagA.\inp\TagB.\End}{\inp\TagA.\out\TagB.\End}$, which
admits only one reduction sequence
\[
    \session{\out\TagA.\inp\TagB.\End}{\inp\TagA.\out\TagB.\End}
    \red
    \session{\inp\TagB.\End}{\out\TagB.\End}
    \red
    \session{\End}{\End}
\]
where every intermediate state can reduce further towards successful
termination. Note that a $\inp\TagB$-labelled transition is immediately enabled
by the session type on the left hand side of the initial composition, but the
matching $\out\TagB$-labelled transition becomes enabled only after the first
reduction.
More interestingly, also the composition
$\session{\out\TagA.\inp\TagB.\End}{\out\TagB.\inp\TagA.\End}$ is
correct. In this case the composition may evolve non-deterministically in two
ways, depending on which output message is consumed first:

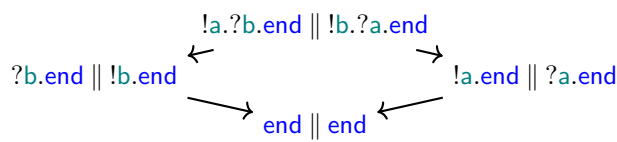
\begin{center}
    \begin{tikzpicture}[thick]
        \node (A)  {$\session{\out\TagA.\inp\TagB.\End}{\out\TagB.\inp\TagA.\End}$};
        \node (BL) [below left=0.5ex of A] {$\session{\inp\TagB.\End}{\out\TagB.\End}$};
        \node (BR) [below right=0.5ex of A] {$\session{\out\TagA.\End}{\inp\TagA.\End}$};
        \node (C)  [below=2em of A] {$\session{\End}{\End}$};
        \draw[->] (A) -- (BL);
        \draw[->] (A) -- (BR);
        \draw[->] (BL) -- (C);
        \draw[->] (BR) -- (C);
    \end{tikzpicture}
\end{center}

In synchronous theories of (binary) session types duality implies correctness,
to the point that in most cases session correctness is expressed in terms of
duality. This does not hold for the notion of correctness given by Bravetti,
Lange and Zavattaro~\cite{BravettiLangeZavattaro24} because they do not consider
the well-formedness condition \eqref{wf:end}. In fact, every type without $\End$
cannot be correctly composed with any other type, including its dual. We now
show that by restricting to types satisfying also the well-formedness condition
\eqref{wf:end} we recover this property which also plays a key role in the
alternative characterization of subtyping discussed in the next section.

\begin{proposition}
    \label{prop:duality-correctness}
    $\correct{\co{S}}S$ holds for every session type $S$.
\end{proposition}

As discussed in \cref{ex:diff_semantics_2}, the presence of FIFO message queues
in \cref{def:transrel} allows types starting with outputs to emit messages and
store them in the queues without performing any check about the possibility for
the receiver to consume these messages. On the contrary, the new reduction
semantics of session composition \eqref{eq:red} checks that all outputs can be
consumed by the partner. Despite this difference, the two semantics can be
proved ``equivalent''. More specifically, we will prove the following
correspondence result: two session types $S$ and $T$ are \emph{compliant}
(according to \cref{def:compatibility}) if and only if $\correct{S}{T}$
(according to \cref{def:correctness}).

In order to state the correspondence result, it is convenient to introduce some
additional notation allowing us to consider the configuration queues as
sequences of input actions to be executed in order to consume the buffered
messages. More precisely, given the queue $\word = \Tag_1 \Tag_2 \dots \Tag_n$,
we write $?\word$ for the corresponding sequence $?\Tag_1 ?\Tag_2 \dots ?\Tag_n$
of input actions. We also use the notation $S \xlred{?\word} S'$ as a shortcut
for $S \xlred{?\Tag_1 ?\Tag_2 \dots ?\Tag_n} S'$. Notice that the queue $\word$
could be empty, i.e. $\word=\epsilon$, in which case $S \xlred{?\word} S$.

In \cref{ex:diff_semantics} we have seen a computation defined according to
\cref{def:transrel} that cannot be mimicked by our reductions defined by
rule~(\ref{eq:red}). We now prove that the vice versa actually holds. More
precisely we identify a way to relate session compositions to session
configurations, and we prove that each reduction of a session compositions can
be mimicked by a corresponding sequence of transitions 
of configurations.
A session composition $\session{S_1}{T_1}$ corresponds to several configurations
$\cnfg{S}{\word_{S}}{T}{\word_{T}}$ such that $S \xlred{?\word_{S}} S_1$ and $T
\xlred{?\word_{T}} T_1$. That is, $S_1$ can be thought of as the residual of $S$
after all the messages in the queue $\word_{S}$ of incoming messages have been
consumed. Similarly for $T_1$ and $T$.\LP{Per caso è una coupled simulation?}
\GZ{Non lo so.. ma per la submission lascerei stare.. pensiamoci dopo}


\begin{example}\label{ex:corresponding}
Consider  
$S = \out\TagA.\inp\TagB.\End$ and $T = \inp\TagA.\out\TagB.\End$. We have
$\session{S}{T} \rightarrow \session{\inp\TagB.\End}{\out\TagB.\End}$. Similarly,
we have the following configuration transition $\cnfg{S}{\epsilon}{T}{\epsilon}
\rightarrow \cnfg{\inp\TagB.\End}{\epsilon}{T}{\TagA}$. Configuration  $\cnfg{\inp\TagB.\End}{\epsilon}{T}{\TagA}$
is one of those corresponding to
$\session{\inp\TagB.\End}{\out\TagB.\End}$ because of the buffered message
$\Tag$ and the type $T$ such that $T \xlred{?\Tag} \out\TagB.\End$. 
Now we have the
reduction $\session{\inp\TagB.\End}{\out\TagB.\End} \rightarrow
\session{\End}{\End}$.
Configuration $\cnfg{\inp\TagB.\End}{\epsilon}{T}{\TagA}$ can mimick this reduction with the two following transitions
$\cnfg{\inp\TagB.\End}{\epsilon}{T}{\TagA} \rightarrow
\cnfg{\inp\TagB.\End}{\epsilon}{\out\TagB.\End}{\epsilon} \rightarrow
\cnfg{\inp\TagB.\End}{\TagB}{\End}{\epsilon}$. The reached configuration
corresponds to $\session{\End}{\End}$ because of the buffered message $\TagB$
and the type $\inp\TagB.\End$ such that $\inp\TagB.\End \xlred{?\TagB}
\End$.\LP{Questo esempio va descritto meglio.}
\MB{Anch'io pensavo lo stesso, ma forse basta semplicemente scambiare la frase che comincia con ``Now we have the reduction...'' con la frase successiva ``The configuration....''.}
\eoe
\end{example}

\begin{restatable}{lemma}{lemmaSessionsToConfig}
\label{lem:sessions_to_config}
Consider the configuration $\cnfg{S}{\word_{S}}{T}{\word_{T}}$ with $S
\xlred{?\word_{S}} S_1$ and $T \xlred{?\word_{T}} T_1$. If $\session{S_1}{T_1}
\rightarrow \session{S_1'}{T_1'}$ then $\cnfg{S}{\word_{S}}{T}{\word_{T}}
\rightarrow^* \cnfg{S'}{\word_{S'}}{T'}{\word_{T'}}$ with $S'
\xlred{?\word_{S'}} S_1'$ and $T' \xlred{?\word_{T'}} T_1'$.
\end{restatable}

We now investigate the possibility for the reduction relation defined by rule
(\ref{eq:red}) to mimick sequences of transitions of corresponding
configurations. This is not true in general, as shown in
\cref{ex:diff_semantics_2}. 
However, we can prove this result under the assumption
that the initial session compositions, or the initial configurations, are
correct.\LP{Ma non è correct in un caso e compliant nell'altro?}
\GZ{le configurazioni sono correct, i tipi dentro le configurazioni
sono compliant ;-)}
 These two cases
are considered in the next two lemmas, the first of which states that a correct
session composition $\session{S}{T}$ can mimick all computations of the
configuration $\cnfg{S}{\epsilon}{T}{\epsilon}$ and the reached
compositions/configurations are related by the  
correspondence relation discussed above and used in \cref{lem:sessions_to_config}.

\begin{restatable}{lemma}{lemmaConfigToSessionsP}
\label{lem:config_to_sessions_P}
Let $\session{S}{T}$ be a correct session composition. If
$\cnfg{S}{\epsilon}{T}{\epsilon} \rightarrow^*
\cnfg{S'}{\word_{S'}}{T'}{\word_{T'}}$ then $\session{S}{T} \rightarrow^*
\session{S_1}{T_1}$ with $S' \xlred{?\word_{S'}} S_1$ and $T'
\xlred{?\word_{T'}} T_1$.
\end{restatable}

A similar correspondence result holds also for initial correct configurations.

\begin{restatable}{lemma}{lemmaConfigToSessionsB}
\label{lem:config_to_sessions_B}
Let $\cnfg{S}{\epsilon}{T}{\epsilon}$ be a correct configuration. Consider the
transition sequence $\cnfg{S}{\epsilon}{T}{\epsilon} \rightarrow^*
\cnfg{S'}{\word_{S'}}{T'}{\word_{T'}}$ and two types $S_1$ and $T_1$ s.t. $S'
\xlred{?\word_{S'}} S_1$ and $T' \xlred{?\word_{T'}} T_1$. If
$\cnfg{S'}{\word_{S'}}{T'}{\word_{T'}} \rightarrow
\cnfg{S''}{\word_{S''}}{T''}{\word_{T''}}$ then $\session{S_1}{T_1}
\rightarrow^* \session{S_1'}{T_1'}$ with $S'' \xlred{?\word_{S''}} S_1'$ and
$T'' \xlred{?\word_{T''}} T_1'$.
\end{restatable}

We can now conclude with the main result of this section,
that is the correspondence between compliance (\cref{def:compliance}) and correctness (\cref{def:correctness}).

\begin{restatable}{theorem}{theoRedTrans}
\label{th:red_trans}
We have that $S$ and $T$ are compliant (\cref{def:compliance})
iff $\correct{S}{T}$
(\cref{def:correctness}).
\end{restatable}

\section{Fair Asynchronous Session Subtyping}
\label{sec:subtyping}

In this section we present the main contribution of the paper, namely a sound
and complete characterization of the refinement relation defined by Bravetti,
Lange and Zavattaro~\cite{BravettiLangeZavattaro24} recalled in
\cref{def:refine}. In light of \cref{th:red_trans}, such refinement can be
alternatively defined using the new labelled transition system for asynchronous
session types defined in \Cref{sec:session-types} and the notion of correctness
in \cref{def:correctness}:
\begin{quote}
  \text{$S \subt T$ if and only if $\correct{R}{T}$ implies $\correct{R}{S}$ for every $R$}
\end{quote}

This definition corresponds to the definition of a subtyping relation for
session following \emph{Liskov's substitution principle}~\cite{LiskovWing94},
where the property to be preserved is session correctness. 

We start the characterization of $\subt$ by introducing the following coiductive
relation which, as we will see later, turns out to be an overapproximation of
$\subt$.

\begin{definition}
  \label{def:csubt}
  We say that $\srel$ is an \emph{asynchronous subtyping relation} if $(S, T)
  \in \srel$ implies:\LP{Tra correctness, compliance, refinement che poi diventa
  fair asynchronous session subtyping, asynchronous subtyping relation, c'è da
  perdersi. Secondo me ci fanno un mazzo così sulla terminologia.}
  \begin{enumerate}
    \item\label{csubt-end} if $T = \End$, then $S = \End$;
    \item\label{csubt-inp} if $T \lred{\inp\Tag} T'$, then $S
    \lred{\inp\Tag} S'$ and $(S', T') \in \srel$;
    \item\label{csubt-out} if $\outputs(T)$, then $\outputs(S)$ and $S
    \lred{\out\Tag} S'$ implies $T \lred{\out\Tag} T'$ and $(S',
    T') \in \srel$.
  \end{enumerate}
\end{definition}

The clauses of \Cref{def:csubt} specify some expected requirements for a session
subtyping relation: a terminated session type $\End$ is in relation with
itself only; every input action in a supertype $T$ should be matched by the same
input action allowed in the subtype $S$ and the corresponding continuations
should still be related by subtyping; dually, every output action allowed by a
subtype $S$ should be matched by an output action allowed by the supertype $T$
and the corresponding continuations should still be related by subtyping.

Interestingly, these clauses are essentially those found in analogous
characterizations of \emph{synchronous subtyping for session
types}~\cite{GayHole05}, modulo the different orientation of $\subt$ due to our
viewpoint based on the substitution of processes rather than on the substitution
of channels.\footnote{The interested reader may refer to Gay~\cite{Gay16} for a
comparison of the two viewpoints.}
However, there are a couple of quirks that separate \Cref{def:csubt} from
sibling definitions.
First of all, recall that transitions like $T \lred{\inp\Tag} T'$ and $S
\lred{\inp\Tag} S'$ may concern ``deep'' input actions enabled by $T$
and $S$, even when $T$ and $S$ begin with output actions. Hence, clause
\eqref{csubt-inp} may allow pairs of session types to be related even if they do
not start with the same kind of actions. A simple instance of this fact is given
by the session types $\out\TagA.\inp\TagB.S$ and $\inp\TagB.\out\TagA.S$
discussed earlier, for which we have $\out\TagA.\inp\TagB.S \lred{\inp\TagB}
\out\TagA.S$ and $\inp\TagB.\out\TagA.S \lred{\inp\TagB} \out\TagA.S$.
Another novelty is that the clauses \eqref{csubt-inp} and \eqref{csubt-out} are
no longer mutually exclusive, since it may happen that $\outputs(T)$ holds for a
$T$ that also allows (deep) input transitions. For example, any asynchronous
subtyping relation that includes the pair $(\out\TagA.\inp\TagB.S,
\out\TagA.\inp\TagB.S)$ must also include the pair $(\inp\TagB.S, \inp\TagB.S)$
because of clause \eqref{csubt-out} as well as the pair $(\out\TagA.S,
\out\TagA.S)$ because of clause~\eqref{csubt-inp}.

\begin{example}
  \label{ex:oracle-batch-types}
  As a first example of a non trivial asynchronous subtyping relation
  we consider the session types $S$ and $T$ presented in \Cref{ex:satellite}
  which formalizes the two possible communication protocols for the
  ground station discussed in \Cref{sec:introduction}.
  We report here the definitions of the types for reader's convenience:
  $S = \out\TagTC.S + \out\TagDONE.S'$ where $S' =
  \inp\TagTM.S' + \inp\TagOVER.\End$ 
  and
  $T =  \inp\TagTM.T + \inp\TagOVER.T'$ where $T' = \out\TagTC.T' + 
  \out\TagDONE.\End$.  
  %
  Both $S$ and $T$ allow sending an arbitrary number of commmands followed by a
  single $\TagDONE$. Both $S$ and $T$ allow receiving an arbitrary number of
  telemetries followed by a single $\TagOVER$. The difference is that in $T$ the
  commands can only be sent after all the telemetries have been received, whereas
  in $S$ the commands can be sent at any time, even before the first telemetry
  has been received.
  To see that $S$ and $T$ can be related by an asynchronous subtyping, observe
  that we have
  \[
    \begin{prooftree}
      \[
        \vdots
        \justifies
        S \xlred{\inp\TagTM} S
      \]
      \[
        \justifies
        S' \xlred{\inp\TagTM} S'
        \using\SyncTrans
      \]
      \justifies
      S \xlred{\inp\TagTM} S
      \using\AsyncTrans
    \end{prooftree}
    \text{\qquad and\qquad}
    \begin{prooftree}
      \[
        \vdots
        \justifies
        S \xlred{\inp\TagOVER} T'
      \]
      \[
        \justifies
        S' \xlred{\inp\TagOVER} \End
        \using\SyncTrans
      \]
      \justifies
      S \xlred{\inp\TagOVER} T'
      \using\AsyncTrans
    \end{prooftree}
  \]
  therefore $\srel \eqdef \set{ (S, T), (T', T'), (\End, \End) }$ is an
  asynchronous subtyping relation.
  %
  At the same time, no asynchronous subtyping relation includes the pair $(T,S)$
  since it violates the clause \eqref{csubt-out}: $\outputs(S)$ holds whereas
  $\outputs(T)$ does not. Indeed, the session type $\co{S}$ relies on those
  early outputs performed by $S$ in order to make progress and
  $\session{\co{S}}{T}$ is stuck.
  \eoe
\end{example}

\begin{example}
  \label{ex:stream-batch-subt}
    Consider again the types $S = \inp\TagReq.\out\TagResp.S +
    \inp\TagStop.\out\TagStop.\End$ and $T = \inp\TagReq.T + \inp\TagStop.T'$
    where $T' = \out\TagResp.T' + \out\TagStop.\End$ respectively modeling the
    protocols of the stream and batch processing servers defined in
    \Cref{ex:stream-batch-types}. To show that $S$ and $T$ are related by an
    asynchronous subtyping relation it is sufficient to show that the relation
    \[
      \srel \eqdef \set{(\out\TagResp^n.S,T), 
      (\out\TagResp^n.\out\TagStop.\End,T') \mid n \in \Nat } \cup \set{ (\End, \End) }
    \]
    is an asycnhronous subtyping because $(S,T) \in \srel$.
    Pairs $(\out\TagResp^n.S,T)$ are necessary to deal with the transition $T
    \xlred{\inp\TagReq} T$ which is matched by $\out\TagResp^n.S
    \xlred{\inp\TagReq} \out\TagResp^{n+1}.S$.
    The pairs $(\out\TagResp^n.\out\TagStop.\End,T')$ account for the transition
    $T \xlred{\inp\TagStop} T'$ which is matched by $\out\TagResp^n.S
    \xlred{\inp\TagReq} \out\TagResp^n.\out\TagStop.\End$.
    Notice that in pairs of the form $(\out\TagResp^{n+1}.\out\TagStop.\End,T')$
    the first type has a transition $\out\TagResp^{n+1}.\out\TagStop.\End
    \xlred{\out\TagResp} \out\TagResp^n.\out\TagStop.\End$ matched by $T'
    \xlred{\out\TagResp} T'$.
    Finally, $\out\TagStop.\End \xlred{\out\TagStop} \End$ is matched by $T'
    \xlred{\out\TagStop} \End$.
    \eoe
\end{example}

We now present a first result relating $\subt$ and asynchronous subtyping. We
have that $\subt$ satisfies the clauses of \Cref{def:csubt} hence it is an
asynchronous sbtyping relation.

\begin{restatable}{theorem}{lemmaSemanticToCoinductive}
    \label{lem:semantic-to-coinductive}
    ${\subt}$ is an asynchronous subtyping relation.
\end{restatable}

As we have anticipated, the largest asynchronous subtyping relation contains
pairs of types that are not in subtyping relation, as illustrated by the next
example.

\begin{example}
  \label{ex:csubt-not-subt}
  Consider two variants of the batch processing server whose responses may be
  positive ($\TagYes$) or negative ($\TagNo$). Their behavior is described by
  the session types $S = \inp\TagReq.\out\TagNo.S +
  \inp\TagStop.\out\TagStop.\End$ and $T = \inp\TagReq.(\out\TagYes.T +
  \out\TagNo.T) + \inp\TagStop.\out\TagStop.\End$. The server behaving as
  $S$ always responds $\TagNo$. The server behaving as $T$ may respond in either
  way.
  Let $T^0 \eqdef T$ and $T^{n+1} \eqdef \out\TagYes.T^n + \out\TagNo.T^n$.
  It is relatively easy to see that $\srel \eqdef \set{ (\out\TagNo^n.S, T^n)
  \mid n\in\Nat } \cup \set{ (\out\TagStop.\End, \out\TagStop.\End),
  (\End, \End) }$ is an asynchronous subtyping relation, and yet $S\!
  \not\subt T$.\LP{Controllare questa $\srel$ mortale.}
  Indeed, consider the session type $R \eqdef \out\TagReq.R'$ where $R' =
  \inp\TagYes.\out\TagStop.\End + \inp\TagNo.R$. Then $\correct{R}{T}$ holds
  but $\correct{R}{S}$ does not. Basically, (a process behaving as) $R$ insists
  on sending requests until it receives a positive response. At that point, it
  is satisfied and quits the interaction. However, only (a process behaving as)
  $T$ is willing to send a positive response, whereas (a process behaving as)
  $S$ is not.
  \eoe
\end{example}

In general, a purely coinductive characterization based on the clauses of
\Cref{def:csubt} allows us to capture the \emph{safety-preserving properties} of
subtyping, those concerning the admissibility of interactions, because they are
supported by an invariant argument. However, the very same clauses do not say
anything about the \emph{liveness-preserving properties} of subtyping, those
concerning the reachability of successful termination, which must be supported
by a well-foundedness argument.
In order to find a sound characterization of subtyping, we have to resort to
\emph{bounded coinduction}~\cite{AnconaDagninoZucca17,Dagnino19}, a particular
case of coinductive definition where we consider the largest relation that
satisfies the clauses in \Cref{def:csubt} and that is also \emph{included in an
inductively defined relation} that preserves the reachability of successful
termination. We dub this inductive relation \emph{convergence} and we denote it
by $\conv$.

Before looking at the formal definition of $\conv$, let us try to acquire a
rough understanding of convergence by recalling \Cref{ex:csubt-not-subt}, in
which we have identified two session types $S$ and $T$ that satisfy the clauses
of \Cref{def:csubt} but are not related by subtyping. In that example, we can
see that the traces that lead $S$ to termination are a subset of the traces that
lead $T$ to termination. This is a general property that holds every time $S$
describes a more deterministic (or ``less demanding'') behavior compared to $T$.
However, in the specific case of \Cref{ex:csubt-not-subt}, the traces that lead
$S$ to termination are \emph{too few}, to the point that there exists a
potential partner (described by $R$ in the example) that terminates solely
relying on those traces of $T$ that have disappeared in $S$.
Contrast $S$ and $T$ those with $S' = \out\TagA.\out\TagA.S' +
\out\TagB.\End$ and $T' = \out\TagA.T' + \out\TagB.\End$. Also in this
case $S'$ is more deterministic than $T'$ and some traces that lead $T'$ to
termination have disappeared in $S'$. Specifically, every trace of the form
$\ActionsA = \out\TagA^{2n+1}\out\TagB$ is allowed by $T'$ but not by $S'$.
However, it is also the case that for each of these traces we can find another
trace $\out\TagA^{2n}\out\TagB$ that shares a common prefix with $\ActionsA$ and
that leads both $S'$ and $T'$ to termination \emph{after an output action}
$\out\TagB$. Since the behavior described by $T'$ is always able to autonomously
veer the interaction towards termination after \emph{any number} of $\TagA$
messages, any session type $R$ that can be correctly combined with $T'$ must be
prepared to receive this $\TagB$ message at any time, meaning that termination
is preserved also if $R$ is combined with $S'$.

Let us now define $\conv$ formally.

\begin{definition}
  \label{def:conv}
  \emph{Convergence} is the relation $\conv$ inductively defined by the rule:
  \[
    \inferrule{
      \forall\ActionsA\in\tr(T)\setminus\tr(S):
      \exists\ActionsB\leq\ActionsA, \TagA:
      S(\ActionsB\out\TagA) \conv T(\ActionsB\out\TagA)
    }{
      S \conv T
    }
  \]
\end{definition}

Intuitively, a derivation for the relation $S \conv T$ measures the
``difference'' between $S$ and $T$ in terms of allowed traces.
When $\tr(T) \subseteq \tr(S)$ there is essentially no difference between $S$
and $T$, so the rule that defines $\conv$ has no premises and turns into an
axiom.
When $\tr(T) \not\subseteq \tr(S)$, then for every trace $\ActionsA$ of actions
allowed by $T$ but not by $S$ there is a prefix $\ActionsB \leq \ActionsA$
shared by both $S$ and $T$ and followed by an output action $\out\MessageType$
that leads to a pair of session types $S(\ActionsB\out\MessageType)$ and
$T(\ActionsB\out\MessageType)$ that are ``slightly less different'' from each
other in terms of allowed traces. Indeed, the derivation for
$S(\ActionsB\out\MessageType) \conv T(\ActionsB\out\MessageType)$ must be
strictly smaller than that for $S \conv T$, or else $S \conv T$ would not be
inductively derivable.

Note that $\conv$ is trivially reflexive, but apart from that it is generally
difficult to understand when two session types are related by convergence
because of the \emph{non-local} flavor of the relation. However, it is easy to
see that any two session types related by an asynchronous subtyping where at
least one of them is finite are also related by convergence.

\begin{proposition}
  \label{prop:finite-conv}
  If\/ $\srel$ is an asynchronous subtyping relation such that $(S, T) \in
  \srel$ and at least one among $S$ and $T$ is finite, then $S \conv T$.
\end{proposition}
\begin{proof}
  A simple induction on either $S$ or $T$, depending on which one is finite.
\end{proof}



When both $S$ and $T$ are infinite, understanding whether $S \conv T$ holds or
not may require some non-trivial reasoning on their traces. Let us work out a
few examples.\LP{Ricontrollare questi esempi con 400 occhi.}

\begin{example}
  Let $S = \out\Tag[a].\out\Tag[a].S + \out\Tag[b].\End$ and $T =
  \out\Tag[a].T + \out\Tag[b].\End$.
  To prove that $S \conv T$, consider $\ActionsA \in \tr(T) \setminus \tr(S)$.
  It must be the case that $\ActionsA = \out\Tag[a]^{2n+1}\out\Tag[b]$ for some
  $n$.
  Take $\ActionsB \eqdef \out\Tag[a]^{2n}$. Now $S(\ActionsB\out\Tag[b]) =
  T(\ActionsB\out\Tag[b]) = \End$, hence $S(\ActionsB\out\Tag[b]) \conv
  T(\ActionsB\out\Tag[b])$ by reflexivity of $\conv$.
  \eoe
\end{example}

\begin{example}
  Consider again $S = \inp\TagReq.\out\TagNo.S +
  \inp\TagStop.\out\TagStop.\End$ and $T = \inp\TagReq.(\out\TagYes.T +
  \out\TagNo.T) \branch \inp\TagStop.\out\TagStop.\End$ from
  \Cref{ex:csubt-not-subt}.
  Since we conjecture $S \not\conv T$ we can focus on a particular $\ActionsA
  \in \tr(T) \setminus \tr(S)$, namely $\ActionsA =
  \inp\TagReq\out\TagYes\ActionsA' \in \tr(T) \setminus \tr(S)$.
  Note that $\outputs(S)$ does not hold, so the only prefix of $\ActionsA$ that
  we may reasonably consider is $\ActionsB \eqdef \inp\TagReq$ and the only
  output action that $S$ may perform after such prefix is $\out\TagNo$.
  But then $S(\ActionsB\out\TagNo) = S$ and $T(\ActionsB\out\TagNo) = T$, hence
  it is not possible to build a well-founded derivation for $S \conv T$.
  We conclude $S \not\conv T$.
  \eoe
\end{example}

\begin{example}
  \label{ex:stream-batch-conv}
  Let us prove that the session types $S$ and $T$ in
  \Cref{ex:stream-batch-types} satisfy the relation $S \conv T$.
  Consider $\ActionsA \in \tr(T) \setminus \tr(S)$. It must be the case that
  $\ActionsA = \inp\TagReq^m\inp\TagStop\out\TagResp^n\out\TagStop$ for some
  $m,n\in\Nat$ with $m \ne n$.
  If $m < n$, then we can take $\ActionsB \eqdef
  \inp\TagReq^m\inp\TagStop\out\TagResp^m \leq \ActionsA$ and now
  $S(\ActionsB\out\TagStop) = T(\ActionsB\out\TagStop) = \End$ and we
  conclude by reflexivity of $\conv$.
  If $m > n$, then we can take $\ActionsB \eqdef
  \inp\TagReq^m\inp\TagStop\out\TagResp^n \leq \ActionsA$ and now
  $S(\ActionsB\out\TagResp) = \out\TagResp^{m-n-1}.\out\TagStop.\End \conv T'
  = T(\ActionsB\out\TagResp)$ using \Cref{prop:finite-conv}
  (because the asynchnronous subtyping relation $\srel$ of 
  \Cref{ex:stream-batch-subt} contains all pairs 
  $(\out\TagResp^n.\out\TagStop.\End,T')$, for every $n$,
  and the type  and the types $\out\TagResp^{m-n-1}.\out\TagStop.\End$
  is finite).
  \eoe
\end{example}

We are finally ready to state our main result
which confirms that $\conv$ indeed provides the
inductively defined space within which we can characterize $\subt$ as the
largest asynchronous subtyping relation. Note that the characterization is both
sound and complete.

\begin{theorem}
  \label{thm:subt}
  $\subt$ is the largest asynchronous subtyping relation included in $\conv$.
\end{theorem}

\begin{example}
  We are finally able to show that $S$ and $T$ in \Cref{ex:stream-batch-types}
  are such that $S \subt T$. In \Cref{ex:stream-batch-subt} we have considered
  the following asynchronous subtyping relation:
  \[
    \srel \eqdef \set{(\out\TagResp^n.S,T), 
    (\out\TagResp^n.\out\TagStop.\End,T') \mid n \in \Nat } \cup \set{ (\End, \End) }
  \]

  We now prove that all the pairs in $\srel$ are also included in $\conv$,
  confirming that $S \subt T$ by \Cref{thm:subt}.
  In \Cref{ex:stream-batch-conv} we have already shown that $S \conv T$. We can
  use similar arguments to show that also the pairs $(\out\TagResp^n.S,T)$ are
  such that  $\out\TagResp^n.S \conv T$. Consider $\ActionsA \in \tr(T)
  \setminus \tr(\out\TagResp^n.S)$. It must be the case that $\ActionsA =
  \inp\TagReq^m\inp\TagStop\out\TagResp^l\out\TagStop$ for some $m,l\in\Nat$
  with $m+n \ne l$.
  If $m+n < l$, then we can take $\ActionsB \eqdef
  \inp\TagReq^m\inp\TagStop\out\TagResp^{m+n} \leq \ActionsA$ and now
  $S(\ActionsB\out\TagStop) = T(\ActionsB\out\TagStop) = \End$ and we
  conclude by reflexivity of $\conv$.
  If $m+n > l$, then we can take $\ActionsB \eqdef
  \inp\TagReq^m\inp\TagStop\out\TagResp^l \leq \ActionsA$ and now
  $S(\ActionsB\out\TagResp) = \out\TagResp^{m+n-l-1}.\out\TagStop.\End \conv T'
  = T(\ActionsB\out\TagResp)$ using \Cref{prop:finite-conv}.
  The pairs $(\out\TagResp^n.\out\TagStop.\End,T')$, $(\out\TagStop.\End,T')$,
  and $(\End, \End)$ all contain at least one finite type, hence they are
  included in $\conv$ by \Cref{prop:finite-conv}.
  \eoe
\end{example}

\section{Related Work and Concluding Remarks}
\label{sec:conclusion}

Gay and Hole~\cite{GayHole05} introduced the first notion of subtyping for
\emph{synchronous} session types. This subtyping supports variance on both
inputs and outputs so that a subtype can have more external nondeterminism (by
adding branches in input choices) and less internal nondeterminism (by removing
branches in output choices). Padovani~\cite{Padovani16} studied a notion of fair
subtyping for \emph{synchronous} multi-party session types that preserves a
notion of correctness similar to our \Cref{def:compliance}. One key difference
between Padovani's fair subtyping and Gay-Hole subtyping is that the variance of
outputs must be limited in such a way that an output branch can be pruned only
if it is not necessary to reach successful termination. Ciccone and
Padovani~\cite{Padovani16,CicconePadovani22A,CicconePadovani22B} have presented
sound and complete charcaterizations of fair subtyping in the synchronous case.
These characterizations all combine a coinductively defined relation and an
inductively defined relation (called ``convergence''), which respectively
capture the safety-preserving and the liveness-preserving aspects of subtyping.

Subtyping relations for asynchronous session types have been defined by Mostrous
and Yoshida~\cite{MostrousY15}, Chen \emph{et
al.}~\cite{ChenDezaniScalasYoshida17} and Bravetti, Lange and
Zavattaro~\cite{BravettiLangeZavattaro24}. In the asynchronous case a subtype
can anticipate output actions w.r.t. its supertype because anticipated outputs
can be stored in the communication queues without altering the overall
interaction behaviour. The first proposal for asynchronous
subtyping~\cite{MostrousY15} allows a subtype to execute loops of anticipated
outputs, thus delaying input actions indefinitely. This could leave orphan
messages in the buffers. The subsequent work of Chen~\emph{et
al.}~\cite{ChenDezaniScalasYoshida17} imposes restrictions on output
anticipation so as to avoid orphan messages. Bravetti, Lange and
Zavattaro~\cite{BravettiLangeZavattaro24} have defined the first fair
(liveness-preserving) asynchronous session subtyping along the lines of
\Cref{def:refine}. However, they only provided a sound (but not complete)
coinductive characterization.

In this paper we present the first complete characterization of the fair
asynchronous subtyping relation defined by Bravetti, Lange and
Zavattaro~\cite{BravettiLangeZavattaro24} using the techniques introduced by
Ciccone and Padovani in the synchronous
case~\cite{Padovani16,CicconePadovani22A,CicconePadovani22B}. In particular, we
complement a coinductively defined subtyping relation with a notion of
convergence. To do so, we leverage a novel ``asynchronous'' operational
semantics of session types in which a type can perform an input transition even
if such input is prefixed by outputs, under the assumption that such input is
present along \emph{every} initial sequence of outputs.

Another fair asynchronous subtyping relation that makes use of a similar
operational semantics has been recently defined by Padovani and
Zavattaro~\cite{padovani2025}. Unlike the relation defined by Bravetti, Lange
and Zavattaro~\cite{BravettiLangeZavattaro24} and characterized in this paper,
the subtyping relation of Padovani and Zavattaro~\cite{padovani2025} preserves a
liveness property that is strictly weaker than successful session termination.
For this reason, their operational semantics is completely symmetric with
respect to the treatment of input and output actions and no inductively-defined
convergence relation is necessary in order to obtain a sound and complete
characterization of subtyping.





A notable difference between the paper of Bravetti, Lange and
Zavattaro~\cite{BravettiLangeZavattaro24} and our own is the fact that we focus
on those session types describing protocols that can always terminate, whereas
Bravetti, Lange and Zavattaro make no such assumption. This choice affects the
properties of the subtyping relation. In particular, Bravetti, Lange and
Zavattaro~\cite{BravettiLangeZavattaro24} show that a supertype can have
additional input branches provided that such branches are \emph{uncontrollable}.
A session type is uncontrollable if there are no session types that can be
correctly composed with it. Our well-formedness condition \eqref{wf:end} in the
definition of types, guarantees that all session types are controllable. We
argue that the characterization presented in this paper can be easily extended
to the more general case where uncontrollable session types are allowed,
following the approach considered by Padovani~\cite{Padovani16}. Basically,
types that do not satisfy the well-formedness condition \eqref{wf:end} can be
\emph{normalized} by pruning away the uncontrollable subtrees. The normalization
yields session types that are semantically equivalent to the original ones and
that satisfy the condition \eqref{wf:end}. At that point, our characterization
can be used to establish whether or not they are related by subtyping.

We envision two lines of development of this work.
One line is concerned with the study of algorithmic versions of our notions of
correct composition, subtyping, and convergence. As for all the known subtypings
for asychronous sessions, these notions turn out to be undecidable
(\Cref{sec:proofs-undecidability}). In the literature, sound (but not complete)
algorithmic characterizations have been investigated for different variants of
asynchronous session subtyping \cite{BravettiCLYZ21,BocchiKM24,
BravettiLangeZavattaro24}. We plan to investigate the possibility to adapt these
approaches to our new formalization of fair asynchronous subtyping, possibly
taking advantage of the novel operational semantics for asynchronous session
types.
Another line of future work concerns the use of fair asynchronous subtyping for
defining type systems ensuring successful session termination of asynchronously
communicating processes. Such type systems have been studied for
\emph{synchronous} communication by Ciccone, Dagnino and
Padovani~\cite{CicconePadovani22A,CicconeDagninoPadovani24}. A first proposal of
such type system for \emph{asynchronous} communication has been recently given
by Padovani and Zavattaro~\cite{padovani2025}. However, as we have pointed out
above, the subtyping relation used in their type system does not preserve
successful session termination. As a consequence, their type system requires a
substantial amount of additional annotations in types and in typing judgements
in order to enforce successful termination. It may be the case that relying on a
subtyping relation that provides stronger guarantees, like the one characterized
in the present paper, could simplify the type system and possibly enlarge the
family of typeable processes.

\bibliographystyle{plainurl}
\bibliography{main,biblio}

\appendix
\section{Supplement to Section~\ref{sec:session-types}}
\label{sec:proofs-session-types}

\propcoindlts*
\begin{proof}
    Recall that every subtree of a session type contains a leaf of the form
    $\End$.
    Let $n$ be the length of the shortest path from the root of $S$ to one of
    its leaves. We proceed by induction on $n$ and by cases on the last rule
    used to derive $S \lred\Action T$.
    
    \proofcase{\SyncTrans}
    We conclude immediately by taking $T \eqdef S$ and $\Actions \eqdef
    \varepsilon$.
    

    \proofcase{\AsyncTrans}
    Then $S = \sum_{i\in I} \out\MessageType_i.S_i$ and $S_i \lred\Action$ for every
    $i\in I$.
    The shortest path from the root of $S$ to one of its leaves must go through
    one of the branches of $S$, say the one with index $k\in I$. Moreover, the
    shortest path from $S_k$ to one of its leaves must have length $n - 1$.
    Using the induction hypothesis we deduce that there exist a $T$ and a string
    $\ActionsB$ made of output labels only such that $S_k \lred\ActionsB T
    \lred\Action$ where the last transition is derived by \SyncTrans.
    We conclude by taking $\ActionsA \eqdef \out\MessageType_k\ActionsB$ and observing
    that $S \lred{\out\MessageType_k} S_k \lred\ActionsB T$.
\end{proof}

\lemmaSessionsToConfig*
\begin{proof}
If $\session{S_1}{T_1} \rightarrow \session{S_1'}{T_1'}$ then
$S_1 \xlred{\alpha} S_1'$ and $T_1 \xlred{\overline\alpha} T_1'$.
We consider the case in which $\alpha =\, !\Tag$
and $\overline\alpha =\, ?\Tag$ (the symmetric case is similar).
If $S_1 \xlred{!\Tag} S_1'$ then $S_1$ starts with outputs.
As $S \xlred{?\word_{S}} S_1$
we have that $S$ possibly starts with inputs but after a prefix
of the inputs in $?\word_{S}$ it will reach a type starting with 
the same outputs of $S_1$. Let $?\word$ be such prefix, and let
$?\word_{S} = ?\word?\word_{S'}$. Let $S''$ be the 
type reached by $S$ after executing the inputs $?\word$.
We have that 
$\cnfg{S}{\word_{S}}{T}{\word_{T}} \rightarrow^*
    \cnfg{S''}{\word_{S'}}{T}{\word_{T}}$.
As $S''$ contains the same initial outputs of $S_1$, we have 
$S'' \xlred{!\Tag} S'$, hence
$\cnfg{S''}{\word_{S'}}{T}{\word_{T}} \rightarrow 
    \cnfg{S'}{\word_{S'}}{T}{\word_{T}\Tag}$.    
We have $S' \xlred{?\word_{S'}} S_1'$ because
$S'' \xlred{?\word_{S'}} S_1$ and $S'$ (resp. $S_1'$)
is the continuations of $S''$ (resp. $S_1$) after
the initial output $!\Tag$. We also have that 
$T' \xlred{?\word_{T}?\Tag} T_1'$ because
$T \xlred{?\word_{T}} T_1$ and $T_1 \xlred{?\Tag} T_1'$.
\end{proof}

\lemmaConfigToSessionsP*
\begin{proof}
The proof is by induction on the length of the
sequence of transitions $\cnfg{S}{\epsilon}{T}{\epsilon} \rightarrow^*   
 \cnfg{S'}{\word_{S'}}{T'}{\word_{T'}}$.

The base case is trivial, because $S'=S$, $T'=T$ and 
$\session{S}{T}$ performs no transition.

Now assume the Lemma holds and 
$\cnfg{S'}{\word_{S'}}{T'}{\word_{T'}} \rightarrow  
 \cnfg{S_1'}{\word_{S_1'}}{T_1'}{\word_{T_1'}}$.
There are four possible cases: $S'$ performs an input,
$S'$ performs an output, $T'$ performs an input, and
$T'$ performs an output.
We consider only the first two cases because the last two
cases are similar symmetric cases.

In the first case we have $\word_{S'} = \Tag\word_{S''}$
and
$\cnfg{S'}{\Tag\word_{S''}}{T'}{\word_{T'}} \rightarrow  
 \cnfg{S''}{\word_{S''}}{T'}{\word_{T'}}$
where $S''$ is the continuation of $S'$ after the input $?\Tag$.
By the induction hypothesis we have $S' \xlred{?\word_{S'}} S_1$.
As $\word_{S'} = \Tag\word_{S''}$, we have
$S' \xlred{?\Tag?\word_{S''}} S_1$ and also
$S'' \xlred{?\word_{S''}} S_1$ because 
$S''$ is the continuation of $S'$ after the input $?\Tag$.
We can conclude that $\session{S_1}{T_1}$ can mimick the
transition simply by performing no reduction (i.e. $\session{S_1}{T_1} \rightarrow^* \session{S_1}{T_1}$)
because $S'' \xlred{?\word_{S''}} S_1$ and, by induction hypothesis,
$T' \xlred{?\word_{T'}} T_1$.

In the second case we have 
$\cnfg{S'}{\word_{S'}}{T'}{\word_{T'}} \rightarrow  
 \cnfg{S''}{\word_{S'}}{T'}{\word_{T'}\Tag}$
where $S''$ is the continuation of $S'$ after the output $!\Tag$.
As $S'$ starts with outputs (including $!\Tag$),
also $S_1$ starts with the same outputs because, by induction hypothesis, 
$S' \xlred{?\word_{S'}} S_1$. Hence $S_1 \xlred{!\Tag} S_1'$.
We have also that $S'' \xlred{?\word_{S'}} S_1'$
because $S''$ (resp $S_1'$) is the continuation of $S'$ (resp. $S_1$)
after the output $!\Tag$.
By induction hypothesis we have $\session{S}{T} \rightarrow^* \session{S_1}{T_1}$. As $\session{S}{T}$ is correct, $\session{S_1}{T_1}$ cannot be stuck. Hence $T_1 \xlred{?\Tag} T_1'$.
By induction hypothesis we have $T' \xlred{?\word_{T'}} T_1$,
hence also $T' \xlred{?\word_{T'}?\Tag} T_1'$.
We can conclude that $\session{S_1}{T_1}$ can mimick the
transition by performing the reduction $\session{S_1}{T_1} \rightarrow
\session{S_1'}{T_1'}$, with 
$S'' \xlred{?\word_{S'}} S_1'$ and
$T' \xlred{?\word_{T'}?\Tag} T_1'$.
\end{proof}

\lemmaConfigToSessionsB*
\begin{proof}
Consider
$\cnfg{S'}{\word_{S'}}{T'}{\word_{T'}} \rightarrow  
 \cnfg{S''}{\word_{S''}}{T''}{\word_{T''}}$.
There are four possible cases: $S'$ performs an input,
$S'$ performs an output, $T'$ performs an input, and
$T'$ performs an output.
We consider only the first two cases because the last two
cases are similar symmetric cases.

In the first case we have $\word_{S'} = \Tag\word_{S''}$
and
$\cnfg{S'}{\Tag\word_{S''}}{T'}{\word_{T'}} \rightarrow  
 \cnfg{S''}{\word_{S''}}{T'}{\word_{T'}}$
where $S''$ is the continuation of $S'$ after the input $?\Tag$.
By hypothesis we have $S' \xlred{?\word_{S'}} S_1$.
As $\word_{S'} = \Tag\word_{S''}$, we have
$S' \xlred{?\Tag?\word_{S''}} S_1$ and also
$S'' \xlred{?\word_{S''}} S_1$ because 
$S''$ is the continuation of $S'$ after the input $?\Tag$.
We can conclude that $\session{S_1}{T_1}$ can mimick the
transition simply by performing no reduction (i.e. $\session{S_1}{T_1} \rightarrow^* \session{S_1}{T_1}$)
because $S'' \xlred{?\word_{S''}} S_1$ and, by hypothesis,
$T' \xlred{?\word_{T'}} T_1$.

In the second case we have 
$\cnfg{S'}{\word_{S'}}{T'}{\word_{T'}} \rightarrow  
 \cnfg{S''}{\word_{S'}}{T'}{\word_{T'}\Tag}$
where $S''$ is the continuation of $S'$ after the output $!\Tag$.
As $S'$ starts with outputs (including $!\Tag$),
also $S_1$ starts with the same outputs because 
$S' \xlred{?\word_{S'}} S_1$. Hence $S_1 \xlred{!\Tag} S_1'$.
We have also that $S'' \xlred{?\word_{S'}} S_1'$
because $S''$ (resp $S_1'$) is the continuation of $S'$ (resp. $S_1$)
after the output $!\Tag$.
It remains to show that $T_1$ can perform the
complementary input $T_1 \xlred{?\Tag} T_1'$.
This is sufficient to close the case because this implies
that $\session{S_1}{T_1}$ can mimick the
transition by performing the reduction $\session{S_1}{T_1} \rightarrow
\session{S_1'}{T_1'}$
with $S'' \xlred{?\word_{S'}} S_1'$ and 
$T' \xlred{?\word_{T'}?\Tag} T_1'$ simply because, by hypothesis, 
$T' \xlred{?\word_{T'}} T_1$ and $T_1 \xlred{?\Tag} T_1'$.
We now prove $T_1 \xlred{?\Tag} T_1'$, by contradiction.
Assume this does not hold. This means that $T_1$ has a sequence 
of outputs which terminates in a type which is either $\End$
or a type starting with inputs that do not include $?\Tag$.
We have that 
$\cnfg{S}{\epsilon}{T}{\epsilon} \rightarrow^*  
 \cnfg{S''}{\word_{S'}}{T'}{\word_{T'}\Tag}$ with
$T' \xlred{?\word_{T'}} T_1$.
This computation can continue with transitions inferred only
by $T'$ which becomes $T_1$ by performing the inputs in 
$?\word_{T'}$ plus outputs possibly prefixing some of these 
inputs. 
In this way we obtain a computation
$\cnfg{S}{\epsilon}{T}{\epsilon} \rightarrow^*  
 \cnfg{S''}{\word_{S''}}{T_1}{\Tag}$. We continue the computation
by letting $T_1$ to execute the sequence 
of outputs which terminates in a type which is either $\End$
or a type starting with inputs that do not include $?\Tag$.
Let $T_1''$ be such type. 
In this way we obtain a computation
$\cnfg{S}{\epsilon}{T}{\epsilon} \rightarrow^*  
 \cnfg{S''}{\word_{S''}}{T_1''}{\Tag}$.
This computation cannot be extended to reach
$\cnfg{\End}{\epsilon}{\End}{\epsilon}$
because $T_1''$ cannot execute any action thus it
cannot consume $\Tag$.
This contradicts the assumption about
$\cnfg{S}{\epsilon}{T}{\epsilon}$ being a correct configuration.
\end{proof}

\theoRedTrans*
\begin{proof}
We start with the only-if direction.
Let $S$ and $T$ be two compliant types.
Then $\cnfg{S}{\epsilon}{T}{\epsilon}$ is a correct configuration.
Consider now a computation 
$\session{S}{T} \rightarrow^* \session{S_1}{T_1}$.
By repeated application of \cref{lem:sessions_to_config}
(one application for each transition)
we have that $\cnfg{S}{\epsilon}{T}{\epsilon}
\rightarrow^* \cnfg{S'}{\word_{S'}}{T'}{\word_{T'}}$
with
$S' \xlred{?\word_{S'}} S_1$ and $T' \xlred{?\word_{T'}} T_1$.
Being $\cnfg{S}{\epsilon}{T}{\epsilon}$ a correct configuration,
we have that $\cnfg{S'}{\word_{S'}}{T'}{\word_{T'}}
\rightarrow^* \cnfg{\End}{\epsilon}{\End}{\epsilon}$.
By repeated application of \cref{lem:config_to_sessions_B}
(one application for each transition),
we have that $\session{S_1}{T_1} \rightarrow^* \session{\End}{\End}$
(because $\End \xlred{\epsilon} \End$). Hence, by \cref{def:correctness}
we conclude $\correct{S}{T}$.

We now move to the if direction.
Let $\correct{S}{T}$.
Consider now a computation 
$\cnfg{S}{\epsilon}{T}{\epsilon} \rightarrow^*
 \cnfg{S'}{\word_{S'}}{T'}{\word_{T'}}$.
By \cref{lem:config_to_sessions_P}
we have that 
$\session{S}{T} \rightarrow^*
 \session{S_1}{T_1}$
with
$S' \xlred{?\word_{S'}} S_1$ and $T' \xlred{?\word_{T'}} T_1$.
Being $\correct{S}{T}$,
we have that $\session{S_1}{T_1} \rightarrow^*
 \session{\End}{\End}$.
By repeated application of \cref{lem:sessions_to_config}
(one application for each transition),
we have that 
$\cnfg{S'}{\word_{S'}}{T'}{\word_{T'}} \rightarrow^*
 \cnfg{S''}{\word_{S''}}{T''}{\word_{T''}}$
with
$S'' \xlred{?\word_{S''}} \End$ and $T'' \xlred{?\word_{T''}} \End$. 
These last two properties imply that $S''$ (resp. $T''$)
is composed by a sequence of inputs that can consume the
messages in the queue $\word_{S''}$ (resp. $\word_{T''}$).
Hence we have that 
$\cnfg{S''}{\word_{S''}}{T''}{\word_{T''}} \rightarrow^*
 \cnfg{\End}{\epsilon}{\End}{\epsilon}$.
By \cref{def:compliance},
we conclude that $\cnfg{S}{\epsilon}{T}{\epsilon}$ is a correct configuration hence $S$ and $T$ are compliant.
\end{proof}

\section{Supplement to \Cref{sec:subtyping}}
\label{sec:proofs-subtyping}

\subsection{Auxiliary Properties of Session Types and Session Correctness}

\begin{proposition}
    \label{prop:input-after-output}
    If $S \lred{\inp\TagA} T$, then $S(\out\TagB) \lred{\inp\TagA} T(\out\TagB)$.
\end{proposition}
\begin{proof}
    It must be the case that $S = \sum_{i\in I} \out\TagB_i.S_i$ and $\TagB =
    \TagB_k$ and $S(\out\TagB) = S_k$ for some $k\in I$, 
    or else $S(\out\TagB)$ would be undefined.
    From the hypothesis $S \lred{\inp\TagA} T$ we deduce that there is a family of
    $T_i$ such that $S_i \lred{\inp\TagA} T_i$ for every $i\in I$ and $T =
    \sum_{i\in I} \out\Tag_i.T_i$.
    We conclude $S(\out\TagB) = S_k \lred{\inp\TagA} T_k = T(\out\TagB)$.
\end{proof}

\begin{proposition}
    \label{prop:outputs-after-input}
    If $S \lred{\inp\TagA} T$, then either $\outputs(S) = \emptyset$ or
    $\outputs(S) = \outputs(T)$.
\end{proposition}
\begin{proof}
    We distinguish two cases depending on the shape of $S$.
    If $S = \sum_{i\in I} \inp\TagA_i.S_i$, then we conclude $\outputs(S) =
    \emptyset$.
    If $S = \sum_{i\in I} \out\TagA_i.S_i$, then from the hypothesis $S
    \lred{\inp\TagA} T$ we deduce that there exists a family of $T_i$ such that
    $S_i \lred{\inp\TagA} T_i$ for every $i\in I$ and $T = \sum_{i\in I}
    \out\TagA_i.T_i$.
    We conclude $\outputs(S) = \set{\TagA_i}_{i\in I} = \outputs(T)$.
\end{proof}


\begin{proposition}
    \label{prop:correctness}
    $\correct{R}{S}$ iff for every $\ActionsA$, $R'$ and $S'$ such that $R
    \lred{\co\ActionsA} R'$ and $S \lred\Actions S'$ we have $\outputs(R')
    \subseteq \inputs(S')$ and $\outputs(S') \subseteq \inputs(R')$ and there
    exists $\ActionsB$ such that $R' \lred{\co\ActionsB} \End$ and $S'
    \lred\ActionsB \End$.
\end{proposition}
\begin{proof}
    Immediate from \Cref{def:correctness}.
\end{proof}
  
\subsection{Auxiliary Relation for Outputs Prefixed by Inputs}

In the proof of soundness of the asynchronous session subtyping relation
we make use of the following relation $\wlred{\out\MessageType}$
representing the presence of output actions that are prefixed
by input actions. Let $\wlred{\out\MessageType}$ be the relation coinductively defined thus:
\begin{equation}
  \label{eq:deep-output}
  \inferrule[\rulename{do1}]{~}{
    S \wlred{\out\MessageType} T
  }
  ~S \lred{\out\MessageType} T
  \qquad
  \inferrule[\rulename{do2}]{
    S_i \wlred{\out\MessageType} T_i
  }{
    \sum_{i\in I} \inp\MessageType_i.S_i \wlred{\out\MessageType} \sum_{i\in I} \inp\MessageType_i.T_i
  }
\end{equation}

We now prove some properties related with this new relation $\wlred{\out\MessageType}$.

\begin{proposition}
    \label{prop:stronger-diamond}
    If $S \lred{\inp\MessageTypeS} S'$ and $S \wlred{\out\MessageTypeT} S''$,
    then there exists $T$ such that $S' \wlred{\out\MessageTypeT} T$ and $S''
    \lred{\inp\MessageTypeS} T$.
\end{proposition}
\begin{proof}
    We reason by cases on the shape of $S$, noting that it cannot be of the form
    $\End$. We only discuss the case $S = \sum_{i\in I}
    \inp\MessageTypeS_i.S_i$, since the other is symmetric.
    Then $\MessageTypeS = \MessageTypeS_k$ and $S' = S_k$ for some $k\in I$.
    From the hypothesis $S \wlred{\out\MessageTypeT} S''$ and \AsyncTrans we
    deduce that there exists a family of $T_i$ such that $S_i
    \wlred{\out\MessageTypeT} T_i$ and $S'' = \sum_{i\in I}
    \inp\MessageTypeS_i.T_i$.
    We conclude by taking $T \eqdef T_k$ and observing that $S''
    \lred{\inp\MessageTypeS} T$.
\end{proof}

\begin{proposition}
    \label{prop:output-after-input}
    If $S \wlred{\out\MessageType} T$, then $S(\inp\MessageTypeT) \wlred{\out\MessageType}
    T(\inp\MessageTypeT)$.
\end{proposition}
\begin{proof}
    We distinguish two cases depending on the shape of $S$.
    \begin{itemize}
      \item ($S = \sum_{i\in I} \inp\MessageType_i.S_i$)
        Then $\MessageTypeT = \MessageType_k$ for some $k\in I$.
        From the hypothesis $S \wlred{\out\MessageType} T$ we deduce that there exists a
        family of $T_i$ such that $S_i \wlred{\out\MessageType} T_i$ for every $i\in I$
        and $T = \sum_{i\in I} \inp\MessageType_i.T_i$.
        We conclude $S(\inp\MessageTypeT) = S_k \wlred{\out\MessageType} T_k = T(\inp\MessageTypeT)$.
      \item ($S = \sum_{i\in I} \out\MessageType_i.S_i$)
        Then $\MessageType = \MessageTypeT_k$ and $T = S_k$ for some $k\in I$.
        From the hypothesis that $S(\inp\MessageTypeT)$ is defined and
        \rulename{l-async} we deduce that there exists a family of $T_i$ such that
        $S_i \lred{\inp\MessageTypeT} T_i$ for every $i\in I$.
        We conclude $S(\inp\MessageTypeT) = \sum_{i\in I} \out\MessageType_i.T_i
        \wlred{\out\MessageType} T_k = S_k(\inp\MessageTypeT) = T(\inp\MessageTypeT)$.
        \qedhere
\end{itemize}
\end{proof}

\begin{proposition}
    \label{prop:inputs-after-output}
    If $S \wlred{\out\MessageType} T$, then $\inputs(S) \subseteq \inputs(T)$.
\end{proposition}
\begin{proof}
    We distinguish two cases depending on the shape of $S$.
    \begin{itemize}
      \item ($S = \sum_{i\in I} \inp\MessageType_i.S_i$)
        From the hypothesis $S \wlred{\out\MessageType} T$ we deduce that there exists a
        family of $T_i$ such that $S_i \wlred{\out\MessageType} T_i$ for every $i\in I$
        and $T = \sum_{i\in I} \inp\MessageType_i.T_i$.
        We conclude $\inputs(S) = \inputs(T)$.
      \item ($S = \sum_{i\in I} \out\MessageType_i.S_i$)
        Then $\MessageType = \MessageType_k$ and $T = S_k$ for some $k\in I$.
        We conclude $\inputs(S) = \bigcap_{i\in I} \inputs(S_i) \subseteq
        \inputs(S_k) = \inputs(T)$.
        \qedhere
    \end{itemize}
\end{proof}
  
\begin{proposition}
    \label{prop:must-output}
    If $\correct{R}{S}$ and $S \wlred{\out\MessageType}$, then there exist $\MessageType_1,
    \dots, \MessageType_n$ and $\Actions$ such that $R
    \xlred{\out\MessageType_1\dots\out\MessageType_n\inp\MessageType\co\Actions} \End$ and $S
    \xlred{\inp\MessageType_1\dots\inp\MessageType_n\out\MessageType\Actions} \End$.
\end{proposition}
\begin{proof}
    From the hypothesis $\correct{R}{S}$ we deduce that there exists $\ActionsB$
    such that $R \lred{\co\ActionsB} \End$ and $S \lred\ActionsB \End$.
    From $S \wlred{\out\MessageType}$ we deduce that $\ActionsB =
    \inp\MessageType_1\cdots\inp\MessageType_n\out\MessageType\ActionsB'$ for some 
    $\MessageType_1$, \dots,
    $\MessageType_n$, $\MessageType$ and $\ActionsB'$.
    Also, $\MessageType \in \outputs(S(\inp\MessageType_1\cdots\inp\MessageType_n)) \cap
    \inputs(R(\out\MessageType_1\cdots\out\MessageType_n))$, therefore we deduce
    $\correct{R(\out\MessageType_1\cdots\out\MessageType_n\inp\MessageType)}
    {S(\inp\MessageType_1\cdots\inp\MessageType_n\out\MessageType)}$.
    We conclude that there exists $\Actions$ such that
    $R(\out\MessageType_1\cdots\out\MessageType_n\inp\MessageType) \lred{\co\Actions} \End$ and
    $S(\inp\MessageType_1\cdots\inp\MessageType_n\out\MessageType) \lred\Actions \End$.
\end{proof}
  
\subsection{Soundness of the Characterization of Subtyping}

\begin{lemma}
    \label{lem:deep-output}
    If\/ $\srel$ is an asynchronous subtyping relation and
    $(S,T) \in \srel$ and $S \lred{\out\MessageType}$, 
    then there exists $T'$ such that
    $T \wlred{\out\MessageType} T'$.
  \end{lemma}
  \begin{proof}%
    Let $\F(S,T)$ be the session type corecursively defined by the following
    equations
    \[
      \F(S,T) =
      \begin{cases}
        T_k
        & \text{if $T = \sum_{i\in I} \out\MessageType_i.T_i$ and $\MessageType = \MessageType_k$ with $k\in I$}
        \\
        \Branch_{i\in I} \inp\MessageType_i.\F(S(\inp\MessageType_i), T_i) &
        \text{if $T = \sum_{i\in I} \inp\MessageType_i.T_i$}
      \end{cases}
    \]
    under the hypotheses that\/ $\srel$ is an asynchronous subtyping relation,
    $(S,T) \in \srel$ and $S \lred{\out\MessageType}$.
    Note that each corecursive application of $\F$ concerns session types that
    satisfy the same hypotheses.
    Now we take $T' \eqdef \F(S,T)$ and we show that $T \wlred{\out\MessageType} T'$.
  
    We apply the coinduction principle to show
    that, under the hypotheses $(S,T) \in \srel$ and $S \lred{\out\MessageType}$, 
    $T \wlred{\out\MessageType} \F(S,T)$ is the conclusion of one of the rules in
    \eqref{eq:deep-output} whose premises also concern session types that satisfy
    the same hypotheses.
    We reason by cases on the shape of $T$.
    \begin{itemize}
      \item $T = \sum_{i\in I} \out\MessageType_i.T_i$ where 
        $\MessageType = \MessageType_k$ with $k\in I$ then $\F(S,T) = T_k$.\\
        Now $T \lred{\out\MessageType} T_k$ and we conclude by observing that $T
        \wlred{\out\Tag} \F(S,T)$ is the conclusion of \rulename{do1} which has no
        premises.
      \item $T = \sum_{i\in I} \inp\MessageType_i.T_i$ then $\F(S,T) = \sum_{i\in
        I} \inp\MessageType_i.\F(S(\inp\MessageType_i),T_i)$.\\
        Now from the hypotheses $(S,T) \in \srel$ and $S \lred{\out\MessageType}$ we deduce
        $(S(\inp\MessageType_i), T_i) \in \srel$ and 
        $S(\inp\MessageType_i) \lred{\out\MessageType}$, for every $i\in I$.
        In the application of the coinductive principle, the latters are 
        the hypotheses related with the conclusions 
        $T_i \wlred{\out\MessageType} \F(S(\inp\MessageType_i), T_i)$, for every $i\in I$.
        We conclude by observing that the latters are the premises for an application
        of \rulename{do2} with conclusion $T \wlred{\out\MessageType} \F(S,T)$.
    \end{itemize}
\end{proof}

\begin{lemma}
    \label{lem:supertype_deep_output}
    Let\/ $\srel$ be the maximal asynchronous subtyping relation included in $\conv$.
    If $(S,T) \in \srel$ and $S \lred{\out\MessageType} S'$ and 
    $T \wlred{\out\MessageType} T'$, then $(S',T') \in \srel$.    
\end{lemma}
\begin{proof}
Let\/ $\srel$ be the maximal asynchronous subtyping relation included in $\conv$.
We consider the relation 
$$
\srel' \eqdef \srel \cup \{(S',T') \mid \
         				   (S,T) \in \srel \wedge 
				            S \lred{\out\MessageType} S' \wedge 
		            	  	T \wlred{\out\MessageType}  T'\}
$$
obtained by extending $\srel$ with the additional pairs $(S',T')$ satisfying
the hypothesis in the statement of the lemma. We then prove that 
$\srel' \subseteq \srel$, thus the thesis follows, i.e.,
the additional pairs in $\srel'$ are also in $\srel$.

In order to prove that $\srel' \subseteq \srel$ we can concentrate
only on the additional pairs $(S',T')$ such that there exist 
$(S,T) \in \srel$ with $S \lred{\out\MessageType} S'$ and 
$T \wlred{\out\MessageType} T'$ because the other pairs $\srel'$ are
in $\srel$ by definition.

We reason by cases on the rule used to derive 
$T \wlred{\out\MessageType} T'$. 

If the rule is $\rulename{do1}$, then
$T \lred{\out\MessageType} T'$, hence also $\outputs(T)$.
Being $\srel$ an asynchronous subtyping, we have that
$(S,T) \in \srel$, $\outputs(T)$, $S \lred{\out\MessageType} S'$, and
$T \lred{\out\MessageType} T'$, imply that also the
additional pair $(S',T') \in \srel$.

If the rule is $\rulename{do2}$ we proceed as follows: we show that 
$\srel'$, which includes the additional pairs
$(S',T')$, is an asynchronous subtyping relation included in $\conv$.
Being $\srel$ the maximal of such relations, it also includes
the additional pairs $(S',T')$.

We first prove that $\srel'$ is an asynchronous subtyping relation,
by showing that all pairs $(S,T) \in \srel'$ satisfy the three clauses of
\Cref{def:csubt}. For all pairs $(S,T)$, excluding the additional pairs $(S',T')$
with $T \wlred{\out\MessageType} T'$ derived from rule $\rulename{do2}$,
the three clauses trivially hold because we already know that $(S,T) \in \srel$
and $\srel$ is an asynchronous subtyping.
We consider now the non trivial pairs: 
$(S',T')$ such that there exist $(S,T) \in \srel$ with
$S \lred{\out\MessageType} S'$ and $T \wlred{\out\MessageType} T'$
derived from rule $\rulename{do2}$.
The latter implies that $T$ starts with inputs hence also
$T'$ starts with inputs, hence $T' \neq \End$ and $\outputs(T')$
does not hold. The unique of the three clauses that apply to the
pair $(S',T')$ is the second one:
\begin{itemize}
\item
$T' \lred{\inp\MessageTypeT} T''$.\\
We have already noticed that $T'$ and $T$ start with inputs.
As $T \wlred{\out\MessageType} T'$ is derived from rule $\rulename{do2}$,
$T' \lred{\inp\MessageTypeT}$ implies that also $T \lred{\inp\MessageTypeT}$
because $T$ and $T'$ have the same initial inputs.
Hence there exists $T'''$ such that $T \lred{\inp\MessageTypeT} T'''$,
moreover, by $\rulename{do2}$, we also have $T''' \wlred{\out\MessageType} T''$.
We now consider $S$. By hypothesis we know that $S \lred{\out\MessageType} S'$. 
Moreover, being $(S,T) \in \srel$, which is an asynchronous subtyping relation,
we also have that $T \lred{\inp\MessageTypeT} T'''$ implies the existence
of $S'''$ such that $S \lred{\inp\MessageTypeT} S'''$ and $(S''',T''') \in \srel$.
Hence we have both $S \lred{\out\MessageType} S'$ and $S \lred{\inp\MessageTypeT} S'''$;
by \Cref{prop:diamond} there exists $S''$ such that 
$S' \lred{\inp\MessageTypeT} S''$ and $S''' \lred{\out\MessageType} S''$.
By $(S''',T''') \in \srel$, $S''' \lred{\out\MessageType} S''$, and
$T''' \wlred{\out\MessageType} T''$ , we conclude that $(S'',T'')$ is one
of the additional pairs in $\srel'$.
So we can conclude that the clause holds because $S' \lred{\inp\MessageTypeT} S''$
and $(S'',T'') \in \srel'$.
\end{itemize}

We now show that $\srel' \subseteq\ \conv$. To do so it is sufficient to
prove that the additional pairs $(S',T')$, such that there exist 
$(S,T) \in \srel$ with $S \lred{\out\MessageType} S'$ and $T \wlred{\out\MessageType} T'$,
are included in $\conv$. In fact, all the other pairs of $\srel'$ are also in $\srel$
hence also in $\conv$, by definition of $\srel$.
Let $(S',T')$ be one of such pairs.
We have that 
\[
\begin{array}{lll}
\tr(S') & = & \{\ActionsA' \ActionsA'' \mid \ActionsA' \out\MessageType \ActionsA'' \in \tr(S)\ \text{and}\ 
				\text{$\ActionsA'$ contains only inputs} \} 
\\
\tr(T') & = & \{\ActionsA' \ActionsA'' \mid \ActionsA' \out\MessageType \ActionsA'' \in \tr(T)\ \text{and}\ 
				\text{$\ActionsA'$ contains only inputs} \} 
\end{array}
\]
Let $\ActionsA\in\tr(T')\setminus\tr(S')$.
We have that $\ActionsA = \ActionsA' \ActionsA''$ with $\ActionsA'$ which contains only
inputs and $\ActionsA' \out\MessageType \ActionsA'' \in \tr(T)$. 
As $\ActionsA \not\in \tr(S')$,
we have $\ActionsA' \out\MessageType \ActionsA'' \not\in \tr(S)$.
$(S,T) \in \srel$ implies that $S \conv T$:
hence $\exists\ActionsB \leq \ActionsA' \out\MessageType \ActionsA''$ and $\TagB$
s.t. $S(\ActionsB\out\TagB) \conv T(\ActionsB\out\TagB)$.
There are two possible cases:
\begin{itemize}
\item
$\ActionsB$ contains at least one output.\\
In this case we have $\ActionsB=\ActionsA' \out\MessageType \ActionsB'$
with $\ActionsB' \leq \ActionsA''$.
We have that $S(\ActionsA' \out\MessageType \ActionsB' \out\TagB) = 
S'(\ActionsA' \ActionsB' \out\TagB)$ as $S'$ has already executed the
output labeled with $\out\MessageType$. The same holds also for $T$ and $T'$,
i.e., $T(\ActionsA' \out\MessageType \ActionsB' \out\TagB) = 
T'(\ActionsA' \ActionsB' \out\TagB)$. Thus
$S(\ActionsB\out\TagB) \conv T(\ActionsB\out\TagB)$ and
$\ActionsB=\ActionsA' \out\MessageType \ActionsB'$ imply that 
$S'(\ActionsA' \ActionsB' \out\TagB) \conv T'(\ActionsA' \ActionsB' \out\TagB)$.
We conclude that there exist $\ActionsA' \ActionsB' \leq \ActionsA$ 
and $\TagB$ such that 
$S'(\ActionsA' \ActionsB' \out\TagB) \conv T'(\ActionsA' \ActionsB' \out\TagB)$.
\item
$\ActionsB$ contains no output.\\ 
In this case we have 
$\ActionsB \leq \ActionsA'$, i.e., $\ActionsA' = \ActionsB \ActionsA'''$
(with also $\ActionsA'''$ containing only inputs).
As $T(\ActionsB\out\TagB)$ is defined, we have $T(\ActionsB) \lred{\out\TagB}$,
hence $\outputs(T(\ActionsB))$.
We have that $(S,T) \in \srel$: being $\srel$ an asynchronous subtyping relation,
by repeated application of the second rule of \Cref{def:csubt} we can conclude that
also $(S(\ActionsB), T(\ActionsB)) \in \srel$ because $\ActionsB$ contains only
inputs. By the third rule of the same definition, $\outputs(T(\ActionsB))$
implies also $\outputs(S(\ActionsB))$.
We now observe that the input transitions $\ActionsB$ are anticipated
deep transitions for $S$, because $S$ starts with outputs, namely, 
$S(\ActionsB) \lred{\out\MessageType} S'$.
Hence the initial output $\out\MessageType$ of $S$ remain
also in $S(\ActionsB)$, i.e., $S(\ActionsB) \lred{\out\MessageType}$.
As $(S(\ActionsB), T(\ActionsB)) \in \srel$, by the third rule of \Cref{def:csubt}
we have that also $T(\ActionsB) \lred{\out\MessageType}$ and
$(S(\ActionsB \out\MessageType), T(\ActionsB \out\MessageType)) \in \srel$.
But $S'$ (resp. $T'$) differs from $S$ (resp. $T$) because
it has already consumed the output $\out\MessageType$, thus 
$S(\ActionsB \out\MessageType)=S'(\ActionsB)$
(resp. $T(\ActionsB \out\MessageType)=T'(\ActionsB)$).
By $(S(\ActionsB \out\MessageType), T(\ActionsB \out\MessageType)) \in \srel$,
$S(\ActionsB \out\MessageType)=S'(\ActionsB)$ and $T(\ActionsB \out\MessageType)=T'(\ActionsB)$,
we conclude $(S'(\ActionsB), T'(\ActionsB)) \in \srel$.
As $\srel \subseteq\ \conv$, $(S'(\ActionsB), T'(\ActionsB)) \in \srel$ 
implies $S'(\ActionsB) \conv T'(\ActionsB)$.
We now observe that 
$\ActionsA''' \ActionsA'' \in\tr(T'(\ActionsB))\setminus\tr(S'(\ActionsB))$;
this holds because $\ActionsB \ActionsA''' \ActionsA'' = \ActionsA \in
\tr(T')\setminus\tr(S')$.
By definition of $\conv$, $S'(\ActionsB) \conv T'(\ActionsB)$ and 
$\ActionsA''' \ActionsA'' \in\tr(T'(\ActionsB))\setminus\tr(S'(\ActionsB))$ 
imply the existence of $\ActionsB' \leq \ActionsA''' \ActionsA''$ and $\TagC$ such that
$S'(\ActionsB\ActionsB'\out\TagC) \conv T'(\ActionsB\ActionsB'\out\TagC)$.
We have that $\ActionsA = \ActionsA' \ActionsA''$, $\ActionsA' = \ActionsB \ActionsA'''$,
and $\ActionsB' \leq \ActionsA''' \ActionsA''$:
these imply $\ActionsB\ActionsB' \leq \ActionsA$.
We then conclude that there exist $\ActionsB\ActionsB' \leq \ActionsA$ 
and $\TagC$ such that 
$S'(\ActionsB\ActionsB' \out\TagC) \conv T'(\ActionsB\ActionsB' \out\TagC)$.
\end{itemize}
As the above reasoning applies to any 
$\ActionsA\in\tr(T')\setminus\tr(S')$, we can conclude $S' \conv T'$.
\end{proof}

\begin{lemma}
    \label{lem:correct-deep-output}
    If $\correct{R}{T}$ and $R \lred{\inp\MessageType} R'$ and $T \wlred{\out\MessageType} T'$,
    then $\correct{R'}{T'}$.
\end{lemma}
\begin{proof}
    We use the alternative characterization of correctness given by
    \Cref{prop:correctness} in order to show that $\correct{R'}{T'}$.
    Assume $R' \lred{\co\Actions} R''$ and $T' \lred\Actions T''$ for some
    $\Actions$.
    We have to prove $\outputs(R'') \subseteq \inputs(T'')$ and $\outputs(T'')
    \subseteq \inputs(R'')$ and that there exists $\ActionsB$ such that $R''
    \lred{\co\ActionsB} \End$ and $T'' \lred\ActionsB \End$.
    We distinguish two possibilities, depending on whether or not $\Actions$
    begins with enough input actions so as to enable the $\out\MessageType$ output from
    $T$.
    \begin{itemize}
      \item ($\Actions = \inp\MessageType_1\dots\inp\MessageType_n\Actions'$ and 
      $T \xlred{\inp\MessageType_1\dots\inp\MessageType_n}\lred{\out\MessageType}$)
      Then we have $R \xlred{\out\MessageType_1\dots\out\MessageType_n\inp\MessageType\co\Actions'} R''$
      and $T \xlred{\inp\MessageType_1\dots\inp\MessageType_n\out\MessageType\Actions'} T''$ and we
      conclude that $R''$ and $T''$ satisfy the desired properties from the
      hypothesis $\correct{R}{T}$.
      \item ($\Actions = \inp\MessageType_1\dots\inp\MessageType_n$ and 
      $T \xlred{\inp\MessageType_1\dots\inp\MessageType_n}\nlred{\out\MessageType}$)
      Then $\outputs(T(\inp\MessageType_1\dots\inp\MessageType_n)) = \emptyset$.
      From the hypothesis $\correct{R}{T}$ we deduce that
      $\outputs(R(\out\MessageType_1\dots\out\MessageType_n)) \ne \emptyset$ and we obtain
      \[
        \begin{array}{r@{~}c@{~}ll}
            \outputs(R'') & = & \outputs(R'(\out\MessageType_1\dots\out\MessageType_n))
            \\
          & \subseteq & \outputs(R(\out\MessageType_1\dots\out\MessageType_n))
          & \text{by \Cref{prop:output-after-input,prop:outputs-after-input}}
          \\
          & \subseteq & \inputs(T(\inp\MessageType_1\dots\inp\MessageType_n))
          & \text{from $\correct{R}{T}$}
          \\
          & \subseteq & \inputs(T'(\inp\MessageType_1\dots\inp\MessageType_n))
          & \text{by \Cref{prop:input-after-output,prop:inputs-after-output}}
          \\
          & = & \inputs(T'')
        \end{array}
      \]

      From $\outputs(T(\inp\MessageType_1\dots\inp\MessageType_n)) = \emptyset$ we also deduce
      $\outputs(T'') = \outputs(T'(\inp\MessageType_1\dots\inp\MessageType_n)) = \emptyset
      \subseteq \inputs(R'')$. 
    \end{itemize}
    Concerning the existence of $\ActionsB$ with the desired properties, this is
    a straightforward consequence of \Cref{prop:must-output}.
\end{proof}
  
\begin{lemma}
    \label{lem:csubt-simulation}
    Let $\srel$ be the maximal asychronous subtyping relation contained
    in $\conv$. If $(S,T) \in \srel$ and $\correct{R}{T}$ and 
    $\session{R}{S} \red \session{R'}{S'}$, then there exists $T'$
    such that $(S',T') \in \srel$ and $\correct{R'}{T'}$.
\end{lemma}
\begin{proof}
    By definition of session reduction we deduce that $\outputs(R) \subseteq
    \inputs(S)$ and $\outputs(S) \subseteq \inputs(R)$ and there exists $\Action
    \in \co{\actions(R)} \cap \actions(S)$ such that $R' = R(\co\Action)$ and
    $S' = S(\Action)$.
    Let us distinguish two subcases depending on the shape of $\Action$.
    \begin{itemize}
        \item ($R \lred{\out\MessageType} R'$ and $S \lred{\inp\MessageType} S'$)
        From the hypothesis $\correct{R}{T}$ we deduce $T \lred{\inp\MessageType} T'$
        where $\correct{R'}{T'}$.
        Being $\srel$ an asynchronous session subtyping relation, from the hypothesis 
        $(S,T) \in \srel$ we conclude $(S',T') \in \srel$ by rule 2 of \Cref{def:csubt}.

        \item ($R \lred{\inp\MessageType} R'$ and $S \lred{\out\MessageType} S'$)
        From \Cref{lem:deep-output} we deduce that there exists $T'$ such that
        $T \wlred{\out\Tag} T'$.
        From \Cref{lem:supertype_deep_output} we deduce that $(S',T') \in \srel$.
        We conclude $\correct{R'}{T'}$ using \Cref{lem:correct-deep-output}.
        \qedhere
    \end{itemize}
\end{proof}

\begin{lemma}
    \label{lem:csubt-termination}
    If $S \conv T$ and $\correct{R}{T}$, then $\session{R}{S} \wred
    \session{\End}{\End}$.
\end{lemma}
\begin{proof}
    We proceed by structural induction on the derivation of $S \conv T$.
    From the hypothesis $\correct{R}{T}$ we deduce $R \lred{\co\Actions} \End$
    and $T \lred\Actions \End$ for some $\Actions \in \tr(T)$.
    If $\Actions \in \tr(S)$, then we conclude immediately $\session{R}{S} \wred
    \session{\End}{\End}$.
    If $\Actions \in \tr(T) \setminus \tr(S)$, then    
    from $S \conv T$ we deduce that there exist $\ActionsB \leq \Actions$ and
    $\Tag$ such that $S(\ActionsB\out\Tag) \conv T(\ActionsB\out\Tag)$, where
    the derivation of this latter relation is strictly smaller than that for $S
    \conv T$.
    From the hypothesis $\correct{R}{T}$ we deduce $R
    \xlred{\co\ActionsB\inp\Tag}$ and also
    $\correct{R(\co\ActionsB\inp\Tag)}{T(\ActionsB\out\Tag)}$.
    Hence we have $S(\ActionsB\out\Tag) \conv T(\ActionsB\out\Tag)$,
    with the derivation of this latter relation strictly smaller than 
    that for $S \conv T$,
    and $\correct{R(\co\ActionsB\inp\Tag)}{T(\ActionsB\out\Tag)}$.
    Using the induction hypothesis we deduce
    $\session{R(\co\ActionsB\inp\Tag)}{S(\ActionsB\out\Tag)} \wred
    \session{\End}{\End}$.
    We conclude $\session{R}{S} \wred
    \session{R(\co\ActionsB\inp\Tag)}{S(\ActionsB\out\Tag)} \wred
    \session{\End}{\End}$.
\end{proof}

\begin{theorem}[soundness]
    \label{thm:asubt-sound}
    Let $\srel$ be the maximal asychronous subtyping relation contained
    in $\conv$. We have that ${\srel} \subseteq {\subt}$.
\end{theorem}
\begin{proof}
    Let $R$ be a session type such that $\correct{R}{T}$. We have to show that for every $S$
    s.t. $(S,T) \in \srel$ then $\correct{R}{S}$.
    Consider a reduction $\session{R}{S} \wred \session{R'}{S'}$.
    By repeated application of \Cref{lem:csubt-simulation}, one application
    for each transition in the sequence of transitions $\session{R}{S} \wred \session{R'}{S'}$, 
    we deduce that there exists $T'$ such that
    $\correct{R'}{T'}$ and $(S',T') \in \srel$. 
    Being $\srel$ included in $\conv$, we have $S' \conv T'$.
    From $S' \conv T'$ and $\correct{R'}{T'}$, 
    by \Cref{lem:csubt-termination} we conclude $\session{R'}{S'} \wred
    \session{\End}{\End}$.
\end{proof}


\subsection{Completeness of the Characterization of Subtyping}

\lemmaSemanticToCoinductive*
\begin{proof}
    Let $\srel \eqdef \set{ (S, T) \mid \forall R: \correct{R}{T} \Rightarrow
    \correct{R}{S} }$. We have to show that $\srel$ satisfies the clauses of
    \Cref{def:csubt}.
    \begin{enumerate}
        \item Suppose $T=\End$.
            $\correct{R}{T}$ implies $R = \End$.
            Then it must be the case that $S = \End$,
            clause 1 of \Cref{def:csubt}.
        \item Suppose $T \lred{\inp\MessageType} T'$.
            Now consider $R'$ such that 
            $\correct{R'}{T'}$. We also have
            $\correct{\out\MessageType.R'}{T}$ because the unique transition
            for $\session{\out\MessageType.R'}{T}$ is 
            $\session{\out\MessageType.R'}{T} \red \session{R'}{T'}$.
            By definition of $\srel$ it must be the case that
            $\correct{\out\MessageType.R'}{S}$, hence $\outputs(\out\MessageType.R') = 
            \set\MessageType \subseteq \inputs(S)$. We deduce $S \lred{\inp\MessageType} S'$ 
            for some $S'$.
            From $\correct{\out\MessageType.R'}{S}$ we also deduce $\correct{R'}{S'}$.
            Since $R'$ is arbitrary, we conclude $(S', T') \in \srel$ by
            definition of $\srel$, as required by clause \ref{csubt-inp} of
            \Cref{def:csubt}.
        \item Suppose $\outputs(T)$, namely $T = \sum_{i\in I}
            \out\MessageType_i.T_i$.
            Let $\set{R_i}_{i\in I}$ be a family of session types such that 
            $\correct{R_i}{T_i}$ for every $i\in I$,
			and consider $R \eqdef \sum_{i\in I} \inp\MessageType_i.R_i$.
            By construction of $R$ we have $\correct{R}{T}$, hence
            $\correct{R}{S}$ by definition of $\srel$.
            Then we deduce $\outputs(S)$ and also $\outputs(S) \subseteq
            \inputs(R)$ or else $\session{R}{S}$ would be stuck.
            Now consider $S \lred{\out\MessageType} S'$. It must be the case that 
            $\MessageType = \MessageType_k$ for some $k\in I$. Also, 
            $T \lred{\out\MessageType} T_k$.
            By construction of $R$ we know $\correct{R_k}{T_k}$ and by 
            $\correct{R}{S}$ and $\session{R}{S} \red \session{R_k}{S'}$,
            we also have $\correct{R_k}{S'}$.
            Since the $R_i$ are arbitrary, we conclude $(S', T_k) \in \srel$ as
            required by clause 3 of \Cref{def:csubt}.
            \qedhere
    \end{enumerate}
\end{proof}

\begin{theorem}[completeness]
    \label{thm:asubt-complete}
	Let $\srel$ be the maximal asychronous subtyping relation contained
    in $\conv$. We have that ${\subt} \subseteq {\srel}$.
\end{theorem}
\begin{proof}
    \newcommand{\disc}[2]{\mathcal{D}(#1,#2)}
    From \Cref{lem:semantic-to-coinductive} we already know that ${\subt}$
    is an asynchronous session subtyping relation hence it is sufficient
    to show that $\subt\ \subseteq\ \conv$. 
    We proceed by contradiction from the hypotheses $S \subt T$ and $S
    \not\conv T$.
    To reach the contradiction, we define a session type $\disc{S}{T}$ called
    ``discriminator'' under the hypotheses $S \subt T$ and $S \not\conv T$. The
    discriminator is corecursively defined by the following equations.
    \[
        \disc{S}{T} =
        \begin{cases}
            \sum_{i\in I, S(\inp\MessageType_i) \not\conv T_i} \out\MessageType_i.\disc{S(\inp\MessageType_i)}{T_i}
            &
            \text{if $T = \sum_{i\in I} \inp\MessageType_i.T_i$}
            \\
            \\
            \sum_{i\in I} \inp\MessageType_i.\disc{S_i}{T_i}
            +
            \sum_{j\in J} \inp\MessageTypeT_j.\co{T_j}
            &
            \text{if 
              $\begin{array}{l}
                S = \sum_{i\in I} \out\MessageType_i.S_i\qquad \mbox{and}\\
                T = \sum_{i\in I} \out\MessageType_i.T_i + \sum_{j\in J} \out\MessageTypeT_j.T_j
                \end{array}$
             }   
        \end{cases}
    \]

    Note that $\disc{S}{T}$ is well defined. In particular:
    \begin{itemize}
        \item The case $S = T = \End$ is ruled out by the hypothesis $S \not\conv T$.
        \item When $T = \Branch_{i\in I} \inp\MessageType_i.T_i$ we know that
            $S(\inp\MessageType_i)$ is defined from the hypothesis $S \subt T$. Also,
            from the hypothesis $S \not\conv T$ we deduce that there exists
            $k\in I$ such that $S(\inp\MessageType_k) \not\conv T_k$, hence the choice
            in the discriminator is not empty. 
        \item When $T = \sum_{i\in I} \out\MessageType_i.T_i$ we observe that,
        	by \Cref{lem:semantic-to-coinductive}, $\subt$ is an
			asynchronous subtyping relation, hence also
            $\outputs(S)$ 
            for the clause
            \ref{csubt-out} of \Cref{def:csubt}. For the same reasons we also
            deduce that $S$ performs in general a subset of the outputs
            performed by $T$, hence the second case in the definition of the
            discriminator exhausts all the possibilities when $T$ performs
            outputs. 
            Moreover, we have that $S_i \subt T_i$ for every $i \in \outputs(S)$
            because for every family
            $\{R_i\}$ such that $\correct{R_i}{T_i}$ we have that 
            $\correct{ \sum_{i\in I}\inp\MessageType_i.R_i}{\sum_{i\in I} \out\MessageType_i.T_i}$,
            hence also 
            $\correct{ \sum_{i\in I}\inp\MessageType_i.R_i}{S}$ because $S \subt T$,
            from which we deduce $\correct{R_i}{S_i}$ for every $i \in \outputs(S)$.
            Finally, for every $i \in \outputs(S)$ we also have $S_i \not\conv T_i$.
            In fact, if $\outputs(S)$ is a singleton $S_i$ is the unique type
            reachable from $S$ (with a transition labeled with $\MessageType_i$), 
            $S \not\conv T$ and $T_i$ is the type reachable from $T$ with a transition 
            labeled with $\MessageType_i$. If $\outputs(S)$ has more than one element,
            the messages $\MessageType_i$ are tags and if there exists $j  \in \outputs(S)$ such 
            that $S_j \conv T_j$, this would imply that for each $\ActionsA\in\tr(T)\setminus\tr(S)$
            there exists the empty sequence and the tag $\MessageType_j$ such that
            $S(\out \MessageType_j) \conv T(\out \MessageType_j)$.
        \item Concerning the fact every subtree of $\disc{S}{T}$ contains an
            $\End$ leaf, note that from the hypothesis $S \not\conv T$ we deduce
            that there exists a trace $\ActionsA \in \tr(T) \setminus \tr(S)$
            that contains the output of a tag permitted by $T$ but not by $S$.
            This trace leads to some subtree $\co{T_j}$ in the definition of
            $\disc{S}{T}$, and this subtree contains an $\End$ leaf. 
    \end{itemize}

    Now it's easy to see that $\correct{\disc{S}{T}}{T}$ but not
    $\correct{\disc{S}{T}}{S}$, which contradicts the hypothesis $S \subt T$.
    Concerning $\correct{\disc{S}{T}}{T}$, observe that the discriminator always
    outputs a subset of tags accepted by $T$ (top equation) and always accepts
    all tags sent by $T$ (bottom equation). The existence of a common path that
    leads to the termination of $\disc{S}{T}$ and $T$ follows by construction of
    $\disc{S}{T}$, as argued earlier when discussing the well-formedness of the
    discriminator. 
    %
    Concerning the fact that $\correct{\disc{S}{T}}{S}$ does \emph{not} hold,
    simply observe that every path leading to an $\End$ leaf in $\disc{S}{T}$
    goes through an input action $\inp\MessageTypeT_j$ for which $S$ does not perform
    the corresponding output action $\out\MessageTypeT_j$.
\end{proof}

\section{Undecidability results}
\label{sec:proofs-undecidability}

\begin{definition}[Queue Machine]\label{def:queuemachines}
  A queue machine $M$ is defined by a six-tuple
  $(Q , \Sigma , \Gamma , \$ , s , \delta )$ where:
  \begin{itemize}
  \item $Q$ is a finite set of states;
  \item $\Sigma \subset \Gamma$ is a finite set denoting the input
    alphabet;
  \item $\Gamma$ is a finite set denoting the queue alphabet (ranged
    over by $A,B,C$);
  \item $\$ \in \Gamma -\Sigma$ is the initial queue symbol;
  \item $s \in Q$ is the start state;
  \item $\delta : Q \times \Gamma \rightarrow Q\times \Gamma ^{*}$ is
    the transition function ($\Gamma ^{*}$ is the set of sequences of
    symbols in $\Gamma$).
  \end{itemize}
Considering a queue machine
$M=(Q , \Sigma , \Gamma , \$ , s , \delta )$,
a {\em configuration} of $M$ is an ordered pair
$(q,\gamma)$ where $q\in Q$ is its {\em current state} and
$\gamma\in\Gamma ^{*}$ is the 
{\em queue}.  The
starting configuration on an input string $x \in \Sigma^*$ is 
$(s , x\$)$, composed of
the start state $s$ and the input $x$ followed by the initial queue symbol $\$$.
The transition relation ($\rightarrow_{M}$) over configurations 
$Q \times \Gamma ^{*}$, leading from a configuration to the next one,
is defined as follows.
For $p,q \in Q$, $A \in \Gamma$, and $\alpha,\gamma \in \Gamma ^{*}$, we have
$(p,A\alpha )\rightarrow _{M}(q,\alpha \gamma)$ whenever
$\delta (p,A)=(q,\gamma)$. Let $\rightarrow _{M}^{*}$ be the reflexive and
transitive closure of $\rightarrow _{M}$.
A machine $M$ accepts an input $x$ if it terminates on input $x$, i.e. it reaches
a blocking configuration with the empty queue (notice that, as the transition 
relation is total, the unique way to terminate is by emptying the queue). 
Formally $x$ is accepted by $M$ if and only if there exists $q \in Q$ such that 
$(s , x\$) \rightarrow _{M}^{*} (q,\varepsilon)$, where $\varepsilon$ is the empty string.
\end{definition}

Since queue machines can deterministically encode Turing machines
(see, e.g.,~\cite{KozenBook}, page~354, solution to exercise~99),
checking the acceptance of $x$ by a queue machine $M$ is an
undecidable problem.

\begin{definition}\label{encoding_queue_machines}
Consider a queue machine $M=(Q , \Sigma , \Gamma , \$ , s , \delta )$ 
and an input $x \in \Sigma^*$. Let $E,E'\not\in\Gamma$ be special symbols 
outside the queue alphabet.
We define the session type $T^M_x$ as follows:
$$
\begin{array}{llll}
T^M_x & = & \out\Tag[X_1]. \cdots . \out\Tag[X_n] . \out\Tag[\$] . \out\Tag[E] . T^M 
	& \mbox{with $x = X_1 \cdots X_n$} \\
T^M & = & \sum_{A \in \Gamma} \inp\Tag[A] . \out\Tag[A] . T^M +  \inp\Tag[E] . \out\Tag[E] . T^M_E \\
T^M_E & = & \sum_{A \in \Gamma} \inp\Tag[A] . \out\Tag[A] . T^M +  \inp\Tag[E] . \out\Tag[E'] . \End \\
\end{array}
$$
and the session types $S^M_p$, for every $p \in Q$, as follows:
$$
\begin{array}{llll}
S^M_p & = & \sum_{A \in \Gamma} \inp\Tag[A] . \out\Tag[B_1]. \cdots . \out\Tag[B_n] . S^M_{q} 
			+ \inp\Tag[E] . \out\Tag[E] . S^M_p +  \inp\Tag[E'] . \End
	& \mbox{with $\delta(p,A) = (q,B_1 \cdots B_n)$} 
\end{array}
$$
\end{definition}

\begin{proposition}\label{from_machines_to_types}
Consider a queue machine $M=(Q , \Sigma , \Gamma , \$ , s , \delta )$
and the transition $(p,\alpha )\rightarrow _{M}(q,\gamma)$,
with $\alpha = A_1 \cdots A_n$ and $\gamma = A_2 \cdots A_n B_1 \cdots B_m$,
because $\delta(p,A_1) = B_1 \cdots B_m$.
Consider now the two types:
$$
\begin{array}{lll}
T  & = & {\out\Tag[A_1]. \cdots . \out\Tag[A_l]. \out\Tag[E] . \out\Tag[A_{l+1}]. \cdots . \out\Tag[A_{n}].T^M} \\
T' & = & {\out\Tag[A_2]. \cdots . \out\Tag[A_l]. \out\Tag[E] . \out\Tag[A_{l+1}]. \cdots . \out\Tag[A_{n}]
.\out\Tag[B_1]. \cdots . \out\Tag[B_m].T^M}
\end{array}
$$
with $1 \leq l \leq n$.
We have that 
$\session{S^M_p}{T} \red^+ \session{S^M_q}{T'}$.\\
Consider now the two types:
$$
\begin{array}{lll}
R  & = & {\out\Tag[E] . \out\Tag[A_1]. \cdots . \out\Tag[A_{n}].T^M} \\
R' & = & {\out\Tag[A_2]. \cdots . \out\Tag[A_n]. \out\Tag[E] . \out\Tag[B_1]. \cdots . \out\Tag[B_m].T^M}
\end{array}
$$
We have that 
$\session{S^M_p}{R} \red^+ \session{S^M_q}{R'}$.
\end{proposition}
\begin{proof}
Direct consequence of the possibile transitions of
$\session{S^M_p}{T}$ and
$\session{S^M_p}{R}$.
\end{proof}


\begin{lemma}\label{undecidability_termination}
Consider a queue machine $M=(Q , \Sigma , \Gamma , \$ , s , \delta )$
and an input string $x \in \Sigma^*$. We have that $M$ accepts $x$ if and only if 
$\session{S^M_s}{T^M_x} \wred \session{\End}{\End}$.
\end{lemma}
\begin{proof}
We first consider the \emph{only if} part.
By definition, $M$ accepts $x$ implies $(s , x\$) \rightarrow _{M}^{*} (q,\varepsilon)$.
By repeated application of \Cref{from_machines_to_types} we have that
$\session{S^M_s}{T^M_x} \wred \session{S^M_q}{\out\Tag[E] . T^M}$.
But we have that 
$\session{S^M_q}{\out\Tag[E] . T^M} \red  
 \session{\out\Tag[E].S^M_q}{T^M} \red
 \session{S^M_q}{\out\Tag[E].T^M_E} \red 
 \session{\out\Tag[E].S^M_q}{T^M_E} \red 
 \session{S^M_q}{\out\Tag[E'].\inp\End} \red 
 \session{\out\End}{\inp\End}$.

We now move to the \emph{if} part, 
assuming that $\session{S^M_s}{T^M_x} \wred \session{\End}{\End}$.
We first observe that the sequence of reductions generated by  
$\session{S^M_s}{T^M_x}$
is deterministic because there are no internal choices in both types 
$S^M_s$ and $T^M_x$.
We proceed by contradiction assuming that $x$ is not accepted by $M$.
We can repeatedly apply \Cref{from_machines_to_types} to prove that
$\session{S^M_s}{T^M_x}$ has an infinite sequence of reductions.
But, as the sequence of reductions of $\session{S^M_s}{T^M_x}$
is deterministic, this contradicts the existence of the finite 
terminating sequence $\session{S^M_s}{T^M_x} \wred \session{\End}{\End}$.
\end{proof}

\begin{theorem}
The problem of checking the correct composition of two types is undecidable.
\end{theorem}
\begin{proof}
The proof is by reduction from the acceptance problem in queue machines.
Consider a queue machine $M=(Q , \Sigma , \Gamma , \$ , s , \delta )$
and an input string $x \in \Sigma^*$. We have that $M$ accepts $x$ if and only if 
$\correct{S^M_s}{T^M_x}$.
In fact, given that $\session{S^M_s}{T^M_x}$ generates a deterministic sequence
of reductions (as observed in the proof of \Cref{undecidability_termination})
we have that $\correct{S^M_s}{T^M_x}$ if and only if 
$\session{S^M_s}{T^M_x} \wred \session{\End}{\End}$.
But, by \Cref{undecidability_termination}, the latter holds if and only if
$M$ accepts $x$.
\end{proof}

\begin{definition}
Consider a queue machine $M=(Q , \Sigma , \Gamma , \$ , s , \delta )$ 
and a string $x \in \Gamma^*$. Let $E,E'\not\in\Gamma$ be special symbols 
outside the queue alphabet.
We consider $\co{T^M}$, the dual of the type $T^M$ defined in \Cref{encoding_queue_machines}.
Namely:
$$
\begin{array}{llll}
\co{T^M} & = & \sum_{A \in \Gamma} \out\Tag[A] . \inp\Tag[A] . \co{T^M} +  
			   \out\Tag[E] . \inp\Tag[E] . \co{T^M_E} \\
\co{T^M_E} & = & \sum_{A \in \Gamma} \out\Tag[A] . \inp\Tag[A] . \co{T^M} +  \out\Tag[E] . \inp\Tag[E'] . \End \\
\end{array}
$$
and the session types $S^M_{p,x}$, for every $p \in Q$, as follows:
$$
\begin{array}{llll}
S^M_{p,x} & = & \out\Tag[X_1]. \cdots . \out\Tag[X_n] . \out\Tag[\$] . \out\Tag[E] .S^M_p 
	& \mbox{with $x = X_1 \cdots X_n$} \\
S^M_p & = &
\sum_{A \in \Gamma} \inp\Tag[A] . \out\Tag[B_1]. \cdots . \out\Tag[B_n] . S^M_{q} 
			+ \inp\Tag[E] . \out\Tag[E] . S^M_p +  \inp\Tag[E'] . \End
	& \mbox{with $\delta(p,A) = (q,B_1 \cdots B_n)$} 
\end{array}
$$
Notice that the types $S^M_p$ are the same as those defined
in \Cref{encoding_queue_machines}.
\end{definition}

\begin{lemma}\label{undecidability_convergence}
Consider a queue machine $M=(Q , \Sigma , \Gamma , \$ , s , \delta )$
and an input string $x \in \Sigma^*$. We have that $S^M_{p,x} \conv \co{T^M}$
if and only if $M$ accepts $x$.
\end{lemma}
\begin{proof}
We first observe that if $S \conv T$ then there exists a trace of $T$ which is
also a trace of $S$. This can be proved by induction on the depth of the 
proof of $S \conv T$. In the base case all the traces of $T$ are also
traces of $S$. In the inductive case, such common trace can be found
by appending the trace of the induction hypothesis to the prefix
$\ActionsB\out\MessageType$ used in the premise of the last application
of the rule in \Cref{def:conv}.

We now prove the \emph{only if} part.
If $S^M_{p,x} \conv \co{T^M}$ we have seen that there exists a 
trace $\Actions$ of $\co{T^M}$ which is also a trace
of $S^M_{p,x}$. All the traces of $\co{T^M}$ are sequences of pairs
of input/output actions on the same tag, excluding the last actions in the trace
which are $\out\Tag[E] \inp\Tag[E] \out\Tag[E] \inp\Tag[E']$.
Traces of this shape executed by $S^M_{p,x}$ corresponds to 
computations of the machine $M$ on the input $x$ (with the
insertion of some actions on the additional tag $E$ which does 
not affect the computation because they simply dequeue/enqueue $E$).
In fact, output actions corresponds to dequeue operations and
the subsequent input actions have the effect of enqueueing the
symbols produced by the corresponding transition in the machine $M$. 
Given that the trace $\Actions$ is finite,
the corresponding computation of $M$ with input $x$ terminates, 
hence $M$ accepts $x$.

We now move to the \emph{if} part. Assume that $M$ accepts $x$.
Hence there exists a trace $\Actions$ of $\co{T^M}$ corresponding
to the computation of the machine $M$ on the input $x$
which is a sequence of pairs
of input/output on the same tag, excluding the last actions in the trace
which are $\out\Tag[E] \inp\Tag[E] \out\Tag[E] \inp\Tag[E']$.
Such trace can be used to prove that $S^M_{p,x} \conv \co{T^M}$.
In fact, we have that $\Actions$ is also a trace of $S^M_{p,x}$,
and the other traces of $\co{T^M}$ have the same structure described
above for $\Actions$, but if they differ from $\Actions$ is only because
after a common prefix $\ActionsB$ they execute a different output tag.
This allows us to apply the rule in \Cref{def:conv}; the application
can be nested and each rule application restrict the domain 
of traces of $\co{T^M}$ diverging from $\Actions$. Let $k$ be the 
number of output actions in $\Actions$; the maximal depth of the proof of 
$S^M_{p,x} \conv \co{T^M}$ will be $k$.
\end{proof}

\begin{theorem}
The problem of checking the convergence of two types is undecidable.
\end{theorem}
\begin{proof}
Corollary of \Cref{undecidability_convergence}.
%
%
%
\end{proof}

\begin{lemma}\label{undecidability_subtyping}
Consider a queue machine $M=(Q , \Sigma , \Gamma , \$ , s , \delta )$
and an input string $x \in \Sigma^*$. We have that $S^M_{p,x} \subt \co{T^M}$
if and only if $M$ accepts $x$.
\end{lemma}
\begin{proof}
We start with the \emph{only if} part. We assume $S^M_{p,x} \subt \co{T^M}$.
We proceed by contradiction and we consider that $M$ does not accept $x$.
By \Cref{undecidability_convergence} we have that $S^M_{p,x} \not\conv \co{T^M}$.
By \Cref{thm:subt} we have that $\subt$ is the largest asynchronous subtyping 
relation included in $\conv$, hence $S^M_{p,x} \not\conv \co{T^M}$ implies
$S^M_{p,x} \not\subt \co{T^M}$. But this contradicts the above assumption.

We now move to the \emph{only if} part. We assume that $M$ accepts $x$.
Hence there exists a terminating computation of $M$:
$(s , x\$) = (q_0 , \gamma_0) \rightarrow_{M}
 (q_1 , \gamma_1) \rightarrow_{M}
 (q_2 , \gamma_2) \rightarrow_{M}
 \cdots \rightarrow_{M}
 (q_m , \gamma_m) \rightarrow_{M}
 (q_{m+1},\gamma_{m+1}) = (q , \varepsilon)$.
Consider the following relation $\srel$:
$$
\begin{array}{ll}
\{\ (\out\Tag[a_1]. \cdots . \out\Tag[a_n] .S^M_{q_i},\co{T^M}),
    (\out\Tag[a_2]. \cdots . \out\Tag[a_n] .S^M_{q_i}, \inp\Tag[a_1] . \co{T^M})  \ \mid \\
\qquad \qquad \qquad \qquad \qquad \qquad \qquad \qquad
  1 \leq i \leq m,\ \Tag[a_1] \cdots \Tag[a_n] = \gamma_i' \Tag[E] \gamma_i'',\
    \gamma_i = \gamma_i' \gamma_i'',\ \gamma_i' \neq \varepsilon\ \}\ \cup \\
\{\ (\out\Tag[E]. \out\Tag[a_1] \cdots . \out\Tag[a_n] .S^M_{q_{i}},\co{T^M}),
    (\out\Tag[a_1]. \cdots . \out\Tag[a_n] .S^M_{q_{i}}, \inp\Tag[E] . \co{T^M_E})  \ \mid \ 
    1 \leq i \leq m+1,\ \Tag[a_1] \cdots \Tag[a_n] = \gamma_i\}\ \cup \\    
\{\ (\out\Tag[E] .S^M_{q_{m+1}},\co{T^M_E}),
    (S^M_{q_{m+1}},\inp\Tag[E'] . \End),
    (\End,\End)\  \} \\
\end{array}
$$
We have that $\srel$ is an asynchronous subtyping 
relation. Moreover, by applying the same reasoning in the \emph{if} part
of the proof of \Cref{undecidability_convergence} we prove that
all the pairs in $\srel$ belong also to the convergence relation.
Hence we have that $S^M_{p,x} \subt \co{T^M}$ because $(S^M_{p,x}, \co{T^M}) \in \srel$.
\end{proof}

\begin{theorem}
The problem of checking the subtyping of two types is undecidable.
\end{theorem}
\begin{proof}
Corollary of \Cref{undecidability_subtyping}.
\end{proof}

\end{document}